\documentclass[12pt]{article}
\usepackage[utf8]{inputenc}
\usepackage{amsmath}

\usepackage{amssymb}
\usepackage{amsthm}
\usepackage{mathtools}
\usepackage{color}
\usepackage{bigints}
\usepackage{setspace}
\usepackage{enumerate}
\usepackage{comment}
\usepackage{natbib}
\usepackage{tikz}
\usepackage{pgfplots}
\usepackage{subcaption}
\usepackage{graphicx,pgfarrows,pgfnodes}
\usepackage{caption}
\usepackage{babel}
\usepackage{array}
\usepackage{verbatim}
\usepackage{indentfirst}

\allowdisplaybreaks
\usepackage[top=1in, bottom=1in, left=1in, right=1in]{geometry}

\usepackage{hyperref}
\newtheorem{axiom}{\sc Axiom}
\newtheorem{defn}{\sc Definition}
\newtheorem{theorem}{\sc Theorem}
\newtheorem{lemma}{\sc Lemma}
\newtheorem{proposition}{\sc Proposition}
\newtheorem{cor}{\sc Corollary}

\title{{\bf Revealed Bayesian Persuasion}%
\footnote{This paper was supported by the Israel Science Foundation 
(grant 2365/25). I would like to thank Henrique de Oliveira, Mark Dean, Piotr 
Dworczak, Yoram Halevy, Emir Kamenica, Elliot Lipnowski, Hannah Trachtman, 
and Roee Levy for useful comments and discussions, as well as various 
participants in seminars at Bar Ilan University, Technion, and Hebrew 
University. All remaining errors are my own.}}

\author{Jeffrey Mensch\footnote{\texttt{jeffrey.mensch@mail.huji.ac.il, https://sites.google.com/site/jeffreyimensch/.}} \\ \centering \it \small Hebrew University of Jerusalem}

\date{\vspace{0.8cm} \today}
\begin{document}

\maketitle
\begin{abstract}

\noindent 

How does one test empirically the hypothesis that a decision maker (DM)
is being influenced by information via Bayesian persuasion? In this paper, 
I consider a DM whose state-dependent preferences are known to an analyst, 
who sees the conditional distribution of choices given the state. I 
provide necessary and sufficient conditions for the dataset to be 
consistent with the DM being Bayesian persuaded by an unobserved sender 
who generates a distribution of signals to ex-ante optimize the sender's 
expected payoff. I thereby provide a tool for empirical work on 
information design.

\vspace{5 mm}


\vspace{20 mm}

\end{abstract}

\pagebreak{}

\section{Introduction}

Since the seminal paper of \cite{Kamenica2011}, a large literature
has developed exploring the topic of Bayesian persuasion, in which a sender
ex-ante optimally sends information to a receiver in order to induce an optimal
distribution of actions. In this paper, I provide necessary and sufficient 
techniques to test whether such a model of interaction is upheld in a dataset,
that is, the data can be explained by preferences of a sender and a 
receiver, where the former is optimally persuading the latter. This 
provides a tool for empirical work to examine the use of information 
design in practice.

The setup I consider is as follows. An analyst observes the distributions
of choices,  as a function of the state, by a decision maker (DM) from 
respective sequence of menus, drawn from a common set of actions. The 
analyst also observes the DM's state-dependent payoffs and the (common) 
prior. The question is then whether the analyst can conclude that the data 
that he observes regarding the DM's choices is consistent with the 
existence of a third party, namely, a ``sender," who is persuading the DM 
based on the sender's own preferences. To make the problem nontrivial, I 
assume that the sender only uses a nontrivial persuasion strategy if 
sending the degenerate distribution of the prior is strictly suboptimal,
i.e., the sender ``benefits from persuasion."

This setup is consistent with two possible environments, both of 
economic interest. The first is where the sender is hidden: i.e., we
do not know of the sender's existence other than through his influence 
on the DM. This interpretation of the model is more consistent with
the reasoning of the literature on stochastic choice, in which one
tries to explain why one is choosing randomly from a given menu.
Interpreted in this way, the model contributes to this literature by
providing an additional testable hypothesis to explain this phenomenon.

An alternative environment in which this setup is relevant is 
where one directly observes a sender choosing a distribution of signals. 
In this case, the conditions of the paper describe a direct check of the 
sender's choice, namely, that he has chosen the optimal distribution of 
signals for his (unobserved) preference. Indeed, much of the literature on 
laboratory experiments on Bayesian persuasion 
\citep{nguyen2017bayesian,au2018bayesian,frechette2022rules,coricelli2023incentives,wu2023competition}
have attempted to test this directly, albeit where the 
sender's preference is known. Little experimental work has been done 
without such structure.

In addition, the non-experimental empirical literature on Bayesian 
persuasion is also quite sparse. \cite{vatter2025quality} empirically 
examines the assignment of quality scores in a context of health insurance 
plans. However, the estimation is parametric, as it assume the functional 
form of the optimization problem, and then estimates these parameters to 
best fit the data. Moreover, unlike Bayesian persuasion, the optimization 
problem for assigning scores is not fully flexible; rather, it partitions
the state space monotonically into convex regions. Similarly, 
\cite{hopenhayn2026optimal} parametrically examine optimal partitions 
for coarse rating design. A major contribution of the present paper
is to allow applied researchers to test for Bayesian persuasion without 
any assumptions about the sender's preferences.

My main result states that the dataset can be rationalized by 
Bayesian persuasion if and only if it satisfies two axioms. The first, 
``No Improving Action Switches" (NIAS), is a standard axiom in the 
literature on information-based choice, first introduced by 
\cite{caplin2015testable}. Roughly, it exploits that, based 
on the pattern of choice, the analyst can infer what the information of 
the DM would be if she chooses a given action at a single posterior only. 
The axiom then states that, at these respective posteriors, the action in 
question is indeed optimal for the DM.

The second axiom, novel to this paper, states as follows. Consider the 
original distributions of posteriors in each respective menu, as inferred
from the state-dependent probabilities of each action. Then there is no
Bayes-plausible way to choose different posteriors such that, on average
across menus, the state-dependent probabilities of each action remains
the same, while placing positive probability at the prior when this was 
not the case originally. The intuition relies on the suboptimality of the 
prior under the hypothesis of Bayesian persuasion: if the prior is 
suboptimal for the sender, then there ought to be no way of maintaining 
the same, optimal payoff while putting positive weight on the prior. 
Analogously to \cite{afriat1967construction}, it is sufficient to consider 
these alternatives, which yield the same original payoffs as a (weighted 
average) across menus, yet are clearly suboptimal after this rearrangement 
of choices.

The axiomatization relies on a characterization by 
\cite{lipnowski2017simplifying} of a minimal set of posteriors that need 
to be considered for Bayesian persuasion, given a menu and known 
preferences of the receiver. They show that one need only consider the 
extreme points of the (convex) regions where each respective action is 
optimal for the receiver. As a result, one need only consider a finite 
number of alternative supports instead of the observed distribution in 
the dataset. This finite characterization allows for the use of the 
standard theorem of the alternative/Farkas' lemma techniques to 
demonstrate the result.

In the case that the sender's preferences are state-independent,
the conditions of the test simplify. Here, the sender does not care about
the correlation between action and state, but only the probabilities that
the respective actions are taken. Therefore, the test drops the 
the state-dependency from the requirement that the probabilities remain 
the same on average across menus. As long as one can preserve such 
probabilities, while placing positive weight on the prior, the dataset is
inconsistent with Bayesian persuasion for a sender with state-independent
preferences (Theorem \ref{siperstest}).

One can use similar methods to analyze sender-optimal Bayesian 
persuasion when all the payoffs are measurable in the posterior mean
\citep{dworczak2019simple}. In this case, the parameters of interest
must be adapted to represent posterior means instead of posteriors. Yet
this is not sufficient, as the characterization of ``benefiting" from
persuasion differs in this environment. I therefore provide necessary and 
sufficient conditions for the sender to benefit from persuasion in
the posterior-mean environment (Proposition \ref{pmben}), as well as 
a minimal set of posterior means needed to be checked (Proposition 
\ref{outmsuff}), thereby providing an analogue to \cite{lipnowski2017simplifying}
for the posterior-mean case. I use this characterization to provide an 
alternative axiom that takes both these conditions and the minimal set
into account. Along with a straightforward modification of NIAS, this
is then necessary and sufficient to rationalize the dataset via
Bayesian persuasion over posterior means (Theorem \ref{PMperstest}).

\section{Related Literature}

The paper builds on related work on stochastic choice in the context 
of information design, specifically that from rational inattention,
which shares many techniques with Bayesian persuasion 
\citep{caplin2013behavioral,caplin2022rationally}. The works of 
\cite{caplin2015revealed}, \cite{chambers2020costly}, 
\cite{denti2022posterior}, and \cite{mensch2026posterior} develop 
revealed preference characterizations of various models of rational 
inattention, in which a DM first acquires information about an unknown 
state, and then chooses optimally conditional on the signal she observes.%
\footnote{I will discuss the relationship between the key axioms in
Section 8.}
Each of these rely on the analyst's access to state-dependent stochastic
choice data and the knowledge of the DM's payoff function. They then use 
the NIAS axiom in combination with one other respective
key axiom to provide necessary and sufficient conditions for their 
characterizations. However, to the present date, there has been no 
revealed preference characterization of Bayesian persuasion.

\cite{doval2024revealed} explores how, in the presence of stochastic 
choice (but not state-dependent) data, one can rationalize the pattern of 
choices as if they arise from a DM who first learns something about the state 
before making their choice. They do so by noting that, if the DM is behaving
in this way, they must (a) have posteriors in the regions where the respective
choices are optimal, and (b) the prior must be a weighted average of the 
posteriors according to the observed probabilities. Since the set of 
priors consistent with this is convex, one can characterize this set by 
test inequalities that characterize its boundary. Thus, testing for the 
hypothesis is equivalent to testing for the inconsistency of the inequalities. 
The work then develops several applications, including to multiple
menus and explicitly constructing information structures that rationalize
the data. However, while some of their results may ``be used to further the
study of Bayesian persuasion,"%
\footnote{Indeed, an earlier title of their paper was ``The Core of Bayesian 
Persuasion.}
it does not explicitly provide any results regarding consistency of a data 
set with a Bayesian persuasion setting.

Closely related is the work by \cite{jakobsen2021axiomatic, 
jakobsen2024temptation}, which provides an axiomatic decision-theoretic
characterization of preferences consistent with a Bayesian persuasion 
representation. In the former paper, he takes as primitives the 
preferences of the sender and the choice correspondence of the receiver, 
and provides an axiomatic characterization of sender preferences over 
information structures that have a Bayesian persuasion representation. 
In the latter, he considers several different potential collections of 
choice or preference primitives, including menu preferences, stochastic 
choice, and choice correspondences.

As in axiomatic decision-theoretic characterizations in general, these 
provide much insight into the economics of the sender's problem in 
Bayesian persuasion, and provide testable implications of the theory. At 
the same time, they are of limited use for incomplete choice data, for 
which the revealed preference approach of the present paper is adapted. 
For instance, when the primitive is stochastic choice, the construction 
of \cite{jakobsen2024temptation} requires the analyst to be able to 
consider \emph{all} possible finite menus of acts involving the finite 
lotteries over the set of outcomes. This allows his analyst to precisely 
pin down, up to an equivalence relation, the preferences of the sender 
and receiver (Theorems 4 and 6). However, it is unclear how to use these 
axioms to test for consistency when the data is incomplete: while the 
axioms are still necessary conditions, and so would still detect a 
direct violation of one of the axioms in the choices of such a menu, it 
is also possible that there could be an \emph{indirect} violation of the 
hypothesis due to inconsistency of the choices from different menus; thus 
the conditions are not demonstrably sufficient for finite data. By 
contrast, as the present paper uses a revealed preference approach, it 
is well-suited to check for such indirect violations via an intuitive 
connection to benefiting from persuasion, by appropriately adapting 
techniques developed by \cite{afriat1967construction}.

\section{Preliminaries}

There are $N$ states $\omega\in\Omega$, and a grand set of actions
$X$. The prior is $p_{0}$. The state-dependent utility for the DM, 
$v(a,\omega)$, is known to the analyst. For any finite menu $A\subset X$, 
the analyst observes the conditional distribution of actions given the state, 
$\sigma_{A}: \Omega\rightarrow \Delta(A)$. Let the set of menus be 
$\mathcal{A}\subset 2^X$.

\begin{defn}
A \emph{state-dependent stochastic choice (SDSC) dataset} is a set 
of pairs $\{(A,\sigma_A)\}_{A\in\mathcal{A}}$.
\end{defn}

The analyst aims to test the hypothesis that the SDSC dataset arises
from the presence of a third party, namely, a \emph{sender}, who
has utility $u: X\times\Omega\rightarrow \mathbb{R}$, and sends
signals to the DM to optimize his payoff ex ante. This would make 
the DM a receiver in a game of Bayesian persuasion. The DM, upon 
receiving a signal realization $s$, chooses the optimal action given
the posterior belief conditional on $s$. It is well known (e.g.,     
\cite{Kamenica2011}, Proposition 1) that it is 
without loss to view signals as recommendations for actions, and 
so we can express the set of signal realizations for menu $A$ as 
$A$ itself. Hence under the hypothesis of Bayesian persuasion, it is 
without loss of optimality for the sender, for each such $a\in A$, to 
induce some associated posterior $p_{A,a}\in \Delta(\Omega)$.%
\footnote{I will shortly explain the connection between the observed data and this hypothetical posterior $p_{A,a}$.} 
As a result, one can write the overall set of recommendations as 
consisting of a distribution $\pi_A\in\Delta(A\times \Delta(\Omega))$ 
over pairs $\{(a,p_{A,a})\}_{a\in A}$ that satisfies Bayes' rule with 
respect to the prior $p_0\in \Delta(\Omega)$.%
\footnote{Note that some of these posteriors may be induced with 
probability $0$ if the action $a$ is not among those that are optimally
recommended.}

In order for the DM to be incentivized to obey this recommendation, 
it is required that 
\begin{equation}
\sum_{\omega\in \Omega} v(a,\omega)p_{A,a}(\omega)\geq  \sum_{\omega\in \Omega} v(b,\omega)p_{A,a}(\omega),\forall b\in A
\end{equation}
The theory standardly assumes that the DM breaks ties in favor of the 
sender; thus, if $\sum_{\omega\in \Omega}u(a,\omega)p(\omega)>\sum_{\omega\in \Omega}u(b,\omega)p(\omega)$ 
and both $a$ and $b$ are optimal for the receiver at $p$, the receiver 
does not choose $b$. One can therefore define the action chosen by the 
receiver, $a^*_A:\Delta(\Omega)\rightarrow A$, and the indirect utility 
function $\phi_A:\Delta(\Omega)\rightarrow\mathbb{R}$ as the expected
payoff that the sender receives given posterior $p$. This will become
useful when comparing the choices in the data with a hypothetically 
optimal choice.

The sender's problem can therefore be written as 
\begin{equation}
\label{obj}
\max_{\pi\in\Delta(A\times \Delta(\Omega))} \sum_{a\in A}\sum_{\omega\in\Omega} u(a,\omega)p_{A,a}(\omega)\pi(a,p_{A,a})
\end{equation}
\[
\mbox{s.t. } \sum_{a\in A} p_{A,a}(\omega)\pi(a,p_{A,a})=p_0(\omega),\,\forall \omega
\]
\[
\sum_{\omega\in \Omega} v(a,\omega)p_{A,a}(\omega)\geq  \sum_{\omega\in \Omega} v(b,\omega)p_{A,a}(\omega),\forall b\in A
\]

Now, one can rationalize this hypothesis for any dataset (assuming 
that $a$ is indeed optimal for the DM conditional on recommendation 
$a$) by having the sender be indifferent between all possible actions
$a\in X$, and therefore all incentive-compatible recommendations are 
optimal. In this case, the Bayesian persuasion problem would become
degenerate. To rule out the possibility that indifference is driving 
the sender's provision of information, I require that the sender
persuade the DM if and only if it is strictly optimal to do so, i.e., for 
optimal signal $\pi_A\in\Delta(A\times \Delta(\Omega))$ such that $\pi_A\neq \delta_{p_0}$,
\begin{equation}
\label{nontrivial}
\sum_{a\in A}\sum_{\omega\in\Omega} u(a,\omega)p_{A,a}(\omega)\pi_A(a,p_{A,a})> \sum_{\omega\in\Omega} u(a^*_A(p_0),\omega)p_0(\omega)
\end{equation}

\begin{defn}
    The sender uses \emph{nontrivial Bayesian persuasion} if he only 
    provides $\pi_A\neq \delta_{p_0}$ when (\ref{nontrivial}) is satisfied,
    and provides $\pi_A=\delta_{p_0}$ otherwise.
\end{defn}

I thus test for the hypothesis that the sender uses nontrivial Bayesian 
persuasion. As shown in \cite{Kamenica2011}, Proposition 2, the case where 
the sender is indifferent between providing information and providing no 
information is nongeneric. Thus the assumption of the sender breaking his 
indifference in this particular way only applies to knife-edge cases.%
\footnote{By contrast, in Proposition 5, \cite{Kamenica2011} show that 
the need for \emph{receiver} tiebreaking is not restricted to a nongeneric 
class of persuasion games.}

To connect the choices of the DM, as represented by the SDSC 
function $\sigma_A$, to the hypothesis that there is a sender who
is using recommendations, one can write the \emph{revealed posterior} 
(Caplin and Martin, 2015) as the conditional distribution of states 
given that a certain action has been chosen. With a slight abuse of 
notation, define
\begin{equation}
\sigma_A(a)\coloneqq \sum_{\omega\in\Omega} \sigma_A(a\vert \omega)p_0(\omega)
\end{equation}

\begin{defn}
For each $a\in A$, the \emph{revealed posterior} $p_{A,a}$ is 
given by 
\begin{equation}
\label{revpos}
p_{A,a}\coloneqq\begin{cases}
\frac{\sigma_A(a\vert \omega)p_0(\omega)}{\sigma_A(a)}, & a\in\mbox{supp}(\sigma_A)\\
p_0, & \mbox{otherwise}
\end{cases}
\end{equation}
\end{defn}
As discussed above, it is without loss for the sender to use recommendation
signals. Therefore, for any optimal $\sigma_A$, there is an optimal 
$\pi_A\in \Delta(A\times \Delta(\Omega))$ over recommendation signals. 
One can then rewrite $\pi_A$ as follows.%
\footnote{This definition differs somewhat from that provided in the 
literature on revealed preference tests for rational inattention 
\citep{denti2022posterior}. It is not sufficient to look only 
at the probability of the signal being chosen, as we shall see that 
the sender's preference for different actions $a,b$ chosen at a given 
posterior within a menu $A$ may differ; when testing for alternative 
distributions, these distinctions may matter. By contrast,
in the case of posterior-separable attention costs, the DM only cares
about the indirect utility, which under the hypothesis is maximized 
by the revealed choice; it therefore does not matter whether multiple 
actions are chosen at the same revealed posterior in the SDSC dataset.}

\begin{defn}
    For each $a\in A$, the \emph{distribution of revealed recommendations}
    $\pi_{\sigma_A}\in \Delta(A\times \Delta(\Omega))$ is given by 
    \begin{equation}
    \label{revdist}
    \pi_{\sigma_A}(a,p)\coloneqq\begin{cases}
        \sigma_A(a), & p=p_{A,a}\\
        0, & \mbox{otherwise}
    \end{cases}
    \end{equation}
\end{defn}

In words, the revealed recommendations are the signals that are imputed
from the action $a\in A$ being chosen precisely when these signal 
realizations recommend to choose $a$. It follows that the revealed
posteriors for $a$ are precisely the posteriors in these recommendations, 
$p_{A,a}$, and the realizations occur with probabilities 
$\pi_{\sigma_{A}}$. Testing the hypothesis of Bayesian persuasion is 
therefore a question of whether one can rationalize $\sigma_A$ as 
deriving from recommendation signals.

\section{Illustrative Example}

In this section, I illustrate the idea that drives the main results of 
the paper through a simple, single-menu example. I then informally
describe how this intuition relates to that of \cite{afriat1967construction}. Finally, I describe how, analogously to 
\cite{afriat1967construction}, the single-menu intuition extends to 
multiple menus.

\noindent\textbf{Example 1:} Consider a binary-state environment 
$\Omega=\{\omega_1,\omega_2\}$ with a menu 
$A=\{a,b,c,d\}$. With some abuse of notation, let $u(\cdot,p)$ 
and $v(x,p)$ be the expected payoffs of the sender and
receiver, respectively, when action $x$ is selected at posterior 
probability $p\in[0,1]$ that $\omega=\omega_2$. Let the payoffs 
be given by 
\[
v(x,p)=\begin{cases}
1-p, & x=a\\
0.9-0.5p, & x=b\\
0.6, & x=c\\
-1+2p, & x=d\\
\end{cases}
\]
\[
u(x,p)=\begin{cases}
0.3-p, & x=a\\
1.2-2p, & x=b\\
2p-0.8, & x=c\\
3p-2.3, & x=d
\end{cases}
\]
The sender's indirect utility, $\phi_A(p)$, is depicted in Figure 
1 below. Suppose that the sender sends a binary signal $\pi_A$
with support on $p\in\{0.4,0.8)$. We draw a line through the respective 
posteriors/indirect utilities, as indicated by the dashed green line in Figure 1, to capture the expected payoff for the sender from this signal.

\begin{figure}[h]
\begin{center}
\resizebox{100mm}{75mm}{
\begin{tikzpicture}
\begin{axis}[
	axis lines=center,
  ymin=-0.02,ymax=1.02,
  xmin=-0.02,xmax=1.02,
  xlabel={$p$}, ylabel={$u$},
  ymajorticks=false,
  xtick = {0.4, 0.5, 0.8}
]
\addplot[blue, line width = 2,samples=200][domain=0:0.2] {0.3-x};
\addplot[blue, line width = 2,samples=200][domain=0.2:0.6] {0.8-2*(x-0.2)};
\addplot[blue, line width = 2,samples=200][domain=0.6:0.8] {0.4+2*(x-0.6)};
\addplot[blue, line width = 2,samples=200][domain=0.8:1] {0.1+3*(x-0.8)};
\addplot[olive,line width = 2,dashed,samples=200]{x};
\draw[olive] (axis cs: 0.6, 0.6) node[above] {$\lambda_A$};
\draw[orange, dashed] (axis cs: 0.5,0) -- (axis cs: 0.5,1);
\node[fill, olive, circle, inner sep=1.5pt] 	at (axis cs:0.4,0.4){};
\node[fill, olive, circle, inner sep=1.5pt] 	at (axis cs:0.8,0.8){};
\end{axis}
\end{tikzpicture}
}

\caption{Signaling $\pi_A$ in menu $A$}
\end{center}
\end{figure}

Now suppose that the sender is a Bayesian persuader; in that case,
the chosen posteriors, and the sender's utility at these points, must
lie along the concavification of the sender's utility function. 
It is clear, then, that $\pi_A$ is suboptimal, as there are 
posteriors for which the indirect utility of the sender lies above the line given by $\lambda_A$. Indeed, it is
not hard to see that the optimal distribution induces the posteriors
$p\in \{0.2,0.8\}$ with equal probability.

Indeed, it turns out that the distribution in Figure 1 is suboptimal 
for \emph{any} sender preferences where the sender benefits from 
persuasion. Under this hypothesis, the prior must give a lower indirect
utility than the expected indirect utility from the posteriors. If so, 
then it shouldn't be possible to achieve the same indirect utility as this
distribution by choosing a different distribution that places positive 
probability on the prior. Yet this is not the case in the present example:
one can, instead of inducing posterior $p=0.4$, induce 
$p=0.2$ with probability $\frac{1}{3}$ and $p=0.5=p_0$
with probability $\frac{2}{3}$. In Figure 2, I depict splitting 
$p=0.4$ in the latter way, as indicated by the red dashed line.

\begin{figure}[h]
\begin{center}
\resizebox{100mm}{75mm}{
\begin{tikzpicture}
\begin{axis}[
	axis lines=center,
  ymin=-0.02,ymax=1.02,
  xmin=-0.02,xmax=1.02,
  xlabel={$p$}, ylabel={$u$},
  ymajorticks=false,
  xtick = {0.2, 0.4, 0.5, 0.8}
]
\addplot[blue, line width = 2,samples=200][domain=0:0.2] {0.3-x};
\addplot[blue, line width = 2,samples=200][domain=0.2:0.6] {0.8-2*(x-0.2)};
\addplot[blue, line width = 2,samples=200][domain=0.6:0.8] {0.4+2*(x-0.6)};
\addplot[blue, line width = 2,samples=200][domain=0.8:1] {0.1+3*(x-0.8)};
\addplot[olive,line width = 2,dashed,samples=200]{x};
\draw[olive] (axis cs: 0.6, 0.6) node[above] {$\lambda_A$};
\draw[orange, dashed] (axis cs: 0.5,0) -- (axis cs: 0.5,1);
\node[fill, olive, circle, inner sep=1.5pt] 	at (axis cs:0.4,0.4){};
\node[fill, olive, circle, inner sep=1.5pt] 	at (axis cs:0.8,0.8){};
\node[fill, red, circle, inner sep=1.5pt] 	at (axis cs:0.2,0.8){};
\node[fill, red, circle, inner sep=1.5pt] 	at (axis cs:0.5,0.2){};
\addplot[red,line width = 2,dashed,samples=200][domain=0.2:0.5] {0.8-2*(x-0.2)};

\end{axis}
\end{tikzpicture}
}

\caption{Equivalent signaling in menu $A$, placing positive weight on $p_0$}
\end{center}
\end{figure}

Notice that this split maintains the same state-dependent
probability of choosing all actions (in particular, $x=b$), while
placing positive probability on the prior $p_0=0.5$. Yet, 
since the prior is suboptimal, this must be compensated at the other 
posterior, $p=0.2$. Thus, $p=0.2$ must give higher indirect utility than 
that along the line connecting the posteriors/indirect utilities at 
$p=0.4$ and $p=0.8$, contradicting the optimality of these posteriors 
for Bayesian persuasion. In my main theorem, I develop a finite set 
of conditions that can be used to check whether such an alternative 
distribution of posteriors is possible.

The force driving the incompatibility of a dataset in a single menu 
parallels the force in \cite{afriat1967construction} for a single 
budget/price.%
\footnote{Afriat-style analyses often assume that the analyst just sees
the price and the choice. With such data, any choice is rationalizable
for a single menu. However, if one strengthens the analyst's observation
to include the budget as well, this is no longer the case.}
For a single menu, one can rule out particular choices due to them being 
\emph{dominated}. In \cite{afriat1967construction}, this is due to them
being contained strictly inside of the budget for the menu; under 
locally nonsatiated preference, this is not possible. This is true 
regardless of what these preferences within the class are.%
\footnote{See, for instance, \cite{nishimura2017comprehensive} for a 
discussion of how to check implications of theories that state that 
certain choices are dominated for all theories in a class.}
Similarly, the hypothesis of the present model states that the prior is 
dominated whenever the sender persuades the DM, regardless of such sender 
preferences. As a result, any revealed distribution that can be replaced, 
as in the example, with one that places the prior in the support is also 
dominated.

To extend to multiple menus, one can similarly take a chapter from 
\cite{afriat1967construction} in considering \emph{all of the implications}
of the theory for preferences that arise indirectly across menus. In 
\cite{afriat1967construction}, these additional implications derive from
the property of \emph{transitivity} of rational preferences. As a result,
it is necessary and sufficient to consider chains of preferences revealed
\emph{within} given menus, which, when appended to each other, form a 
cycle. This can lead to an indirect contradiction of the theory, as this 
constructed preference would lead to a violation of transitivity.

Analogously here, one can expand the set of implications of the theory
by considering multiple menus, because there can be additional ways
to construct a contradiction of the theory for the same factors as in 
Example 1. In particular, aside from direct contradictions from a single 
menu, one can check for indirect contradictions by preserving the same 
the same state-dependent distributions of choices, \emph{on average}, across menus. In turn, this preserves, on 
average, the expected utility. If such an 
alternative distribution results in placing positive weight on the prior 
where there was not such weight originally in the revealed distribution, 
then this contradicts the theory in the same way as in Example 1. Thus, 
since the prior was suboptimal, as indicated by the hypothesis that the 
sender only persuades when it is strictly optimal to do so, the ability to 
get the same expected utility by placing positive weight at the prior 
implies that the sender could obtain higher utility by a different signal 
in the original persuasion problem. So, it is not possible that the sender is 
optimally persuading the DM. I spell this out in more detail in the 
discussion surrounding Theorem 1 in Section 5.

\section{Rationalization of Bayesian Persuasion}

As informally illustrated in Example 1, in order to rationalize the SDSC 
dataset, two conditions must be met. First, the choices of acts must be 
optimal from the perspective of the DM, given the revealed posterior. 
Second, there must not be an alternative way of distributing posteriors
that achieves a greater payoff for the sender. I will capture 
these conditions in two axioms.

\subsection{Axiomatization}

The first axiom, ``No Improving Action Switches" (NIAS), has been
a staple of the revealed preference literature since 
\cite{caplin2015testable} for decisions that the DM takes based on
information that she receives. 

\begin{axiom}[No Improving Action Switches (NIAS)]
\label{NIAS}
For all $A\in\mathcal{A}$ and $a\in\mbox{supp}(\sigma_A)$,
$b\in A$,
\begin{equation}
\sum_{\omega\in\Omega} v(a,\omega)p_{A,a}(\omega)\geq \sum_{\omega\in\Omega} v(b,\omega)p_{A,a}(\omega)
\end{equation}
\end{axiom}

Informally, this axiom states that the DM chooses optimally
at each posterior $p_{A,a}$ that appears in the dataset. So,
if the revealed posterior for action $a$ is $p_{A,a}$, it must
indeed be optimal for the DM to be choosing $a$.

The second axiom relies on a property of Bayesian persuasion,
pointed out by \cite{lipnowski2017simplifying}, that restricts 
the set of posteriors that suffices to be considered in order to find 
the optimal persuasion strategy. Notably, this restriction 
\emph{does not} depend on the sender's preferences.
Thus, without knowing yet what the sender's preferences are,
it is possible to restrict the possible deviations from $\sigma_A$
to this set in order to test whether the sender chooses information
optimally. In the sequel, I follow Lipnowski and Mathevet's
terminology in defining this set and its properties.

The crucial insight is that, for any menu $A$ and each $a\in A$, 
the set of posteriors $p$ at which $a$ is optimal for the DM
is a finite convex polytope. One can thus divide the space of 
posteriors $\Delta(\Omega)$ into these convex sets, each of 
which is described by their extreme points. Since any point in
a convex set can be expressed as a convex combination of 
the extreme points, and payoffs are linear in probabilities
within each set, it follows that an optimal persuasion strategy 
can be found by restricting consideration to these extreme points.

\begin{defn}
\label{poscov}
For any menu $A$, the \emph{posterior cover} $\mathcal{C}_A$
is a family of closed convex subsets $C_{A,a}\subset \Delta(\Omega)$
such that, for each $a,b\in A$, 
\begin{equation}
\sum_{\omega\in \Omega} v(a,\omega)p(\omega)\geq \sum_{\omega\in \Omega} v(b,\omega)p(\omega),\,\forall p\in C_{A,a}
\end{equation}
and 
$$\Delta(\Omega)=\bigcup_{a\in A} C_{A,a}$$
\end{defn}

Thus the posterior cover defines the regions of posteriors over
which $a$ is optimal for the receiver. 

\begin{figure}
\begin{center}
\resizebox{70mm}{60mm}{
\begin{tikzpicture}
    \draw[thick] (0,0) -- (10,0);
    \draw[thick] (0,0) -- (5, 7.071);
    \draw[thick] (10,0) -- (5, 7.071);
    \draw[dashed, blue]  (5,7.071) -- (0, 0);
    \draw[dashed, blue]  (5,7.071) -- (10, 0);
    \draw[dashed, blue]  (7, 4.2426) -- (3, 4.2426);
    \draw[dashed, blue]  (7, 4.2426) -- (7, 0);
    \draw[dashed, blue]  (2, 2.8284) -- (3, 0);
    \draw[dashed, blue] (0,0) -- (10,0);
    \draw[dashed, blue] (6, 4.2426) -- (2.6, 1.13136);
    \draw (5,5) node {$a_1$};
    \draw (1.5,1) node {$a_2$};
    \draw (3.3,3.3) node {$a_3$};
    \draw (5.2,1.3) node {$a_4$};
    \draw (8,1.5) node {$a_5$};
    \draw[fill, red] (0,0) circle (0.1);
    \draw[fill, red] (10,0) circle (0.1);
    \draw[fill, red] (3,0) circle (0.1);
    \draw[fill, red] (7,0) circle (0.1);
    \draw[fill, red] (5,7.071) circle (0.1);
    \draw[fill, red] (7, 4.2426) circle (0.1);
    \draw[fill, red] (3, 4.2426) circle (0.1);
    \draw[fill, red] (2, 2.8284) circle (0.1);
    \draw[fill, red] (6, 4.2426) circle (0.1);
    \draw[fill, red] (2.6, 1.13136) circle (0.1);
\end{tikzpicture}
}

\caption{Posterior cover with outer points}
\end{center}
\end{figure}

\begin{defn}
Given posterior cover $\mathcal{C}_A$, let $\mbox{ext}(C_{A,a})$ be the set
of extreme points of $C_{A,a}$. The \emph{outer points}
of $\mathcal{C}_A$ are then given by
\[
\mbox{out}(\mathcal{C}_A)\coloneqq \{p\in \Delta(\Omega):\,\exists a\in A\mbox{ s.t. }p\in\mbox{ext}(C_{A,a})\}
\]
\end{defn}

\noindent\textbf{Example 2:} In Figure 3, I depict an example of a 
posterior cover for a case in which there are $3$ possible states, and the 
receiver can choose between $5$ possible actions 
$A=\{a_1,a_2,a_3,a_4,a_5\}$. The respective regions in which the 
actions are optimal compose the posterior cover; the outer points are 
depicted in red.

\begin{proposition}[\cite{lipnowski2017simplifying}, Proposition 1]
\label{LM2019}
For any sender preferences $u$, any posteriors in an optimal persuasion 
strategy can be replaced with a mean-preserving spread with support on 
$\mbox{out}(\mathcal{C}_A)$.
\end{proposition}

\begin{cor}
\label{outsuff}
For any sender preferences $u$ and menu $A$, there is an optimal 
persuasion strategy that has support on a set of pairs 
$\{(a,p):a\in A, p\in \mbox{ext}(C_{A,a})\}$.
\end{cor}

\begin{proof}
    Follows immediately from the fact that $\mbox{out}(\mathcal{C}_A)=\bigcup \mbox{ext}(C_{A,a})$.
\end{proof} 

It is now possible to construct an axiom for a Bayesian persuasion
representation by reference to the set 
$\{(a,p)\}_{a\in A, p\in \mbox{ext}(C_{A,a})}$. The idea driving the
axiom is analogous to that of \cite{afriat1967construction}, which
considers a dataset comprising choices and prices, 
$(x_i,p_i)\in \mathbb{R}^{2N}$ in $M$ menus. In the next three paragraphs,
I give an overview of how the construction of \cite{afriat1967construction}
works, and how it connects to the results here.

\cite{afriat1967construction} shows that the choice data is consistent 
with utility maximization subject to a budget constraint if and only if 
the data satisfy the \emph{generalized axiom of revealed preference} 
(GARP). GARP is an acyclicality condition that ensures that one cannot 
switch around the choices that the DM makes between menus. Thus, if we 
view these cycles as perturbing the choices observed in the dataset, GARP 
asks whether one can perturb the choices in such a way that, on average, 
keeps them the same across menus. By doing so, the DM's utility across 
menus must remain the same: after all, she is choosing the same bundles on 
average. However, if after switching the choices between menus, one of the 
choices in menu $i$, $x_i$, is now strictly inside the budget set $i+1$, 
i.e., $x_i\cdot p_{i+1}<x_{i+1}\cdot p_{i+1}$, the choice $x_i$ must be
strictly suboptimal under the hypothesis of locally nonsatiated utility. 
As a result, since this perturbation is feasible and preserves the same
payoffs, yet yields suboptimal choices, the original choices must also
be suboptimal: namely, there is some menu $i$ in which $x_i$ is not the
optimal choice. 

Conversely, suppose that no such perturbation is feasible. If one can 
define monotone preferences such that, with respect to them, no 
improvement is possible, then one has rationalized the data. By the 
appropriate use of duality techniques, \cite{afriat1967construction} shows 
that GARP is also sufficient to rationalize the dataset.

In a similar manner, I consider perturbations of the choices of the 
DM at each menus that, \emph{on average}, preserves the expected utility 
of the DM. These are precisely the ones that, while satisfying Bayes' 
rule, preserve the average state-dependent probability of each $a\in X$ 
across menus $\{A_i\}_{i=1}^M$. If the perturbed signal is suboptimal, 
then we have a contradiction to the optimality of the original signal, 
the same way as keeping the same choices across menus in 
\cite{afriat1967construction} leads to a refutation of that model. By the 
hypothesis of nontrivial Bayesian persuasion, the prior $p_0$ is a 
suboptimal posterior for menu $i$ when, in the dataset, it is not the only 
revealed posterior for that menu. So, if one preserves the expected 
utility through a perturbation, but as a result has positive weight on the 
prior in menu $i$, the original choice at menu $i$ must be suboptimal. 
Conversely, if a finite set of linear inequalities guarantees that no such 
perturbation exists, then one can use duality techniques to ensure that 
there is a preference that rationalizes the dataset. 

By the hypothesis of the model, the sender only sends information if
he benefits from persuasion. Therefore, if the alternative choices of
signals put positive weight on a posterior equal to the prior, the choice
is suboptimal. So, if the perturbation preserves the state dependent 
probabilities of each action, and hence the expected utility, the model
is falsified by the above argument.%
\footnote{The theory is immediately falsified if both $p_0$
and some other $p\neq p_0$ are revealed posteriors for menu $A$,
so I ignore this case.}
Conversely, to rationalize the model, one needs a sufficient set of 
perturbations that cover any possible alternative signal that could 
be better for the sender. As argued in Corollary \ref{outsuff}, the set 
$\mbox{out}(\mathcal{C}_A)$ is sufficient for the posteriors in 
optimal persuasion strategies. In addition, one must also consider the
posterior $p_0$ for actions that contain $p_0\in C_{A,a}$ to check for 
the possibility of preserving the state-contingent probabilities of 
each action, while placing positive weight on the prior. Hence I define 
the set 
\begin{equation}
P_{A,a}\coloneqq \begin{cases}
\mbox{ext}(C_{A,a}), & a\notin \arg\max_{\hat{a}\in A} \sum_{\omega\in \Omega} v(\hat{a},\omega)p_0(\omega)\mbox{ or } p_{A,a}=p_0\\
\mbox{ext}(C_{A,a})\cup\{p_0\}, & a\in \arg\max_{\hat{a}\in A} \sum_{\omega\in \Omega} v(\hat{a},\omega)p_0(\omega)\mbox{ and } p_{A,a}\neq p_0
\end{cases}
\end{equation}

I next define a sufficient set of signals that need be considered 
as alternatives to the revealed posteriors, preserving the state-dependent
probabilities of each action.

\begin{defn}
    For any $A$, a finite signal distribution 
    $\pi\in \Delta(A\times \Delta(\Omega))$ is a \emph{prior/outer point 
    (POP)} signal if, for all 
    $a\in A,\,p\in \Delta(\Omega)$, 
    \[
    \pi(a,p)\neq 0\implies p\in P_{A,a}
    \]
\end{defn}

In words, a POP signal has realizations $(a,p)$, which tell
the receiver both (i) the action $a$ to choose, and (ii) the 
posterior $p$ conditional on the signal. By restricting to 
$p\in P_{A,a}$, the action $a$ is indeed optimal for the receiver
conditional on $p$.

\begin{comment}
\begin{axiom}[No Act-Preserving Posterior Cycles (NBPS)]
\label{NBPS}
There does not exist 
$\hat{\beta}\in\mathbb{R}_+^{\sum_{A\in \mathcal{A},a\in A}\vert P_{A,a}\vert}$ 
and $\beta\in\mathbb{R}_+^{\sum_{A\in \mathcal{A},a\in A}\vert \mbox{ext}(p_{A,a})\vert}$ 
such that
\begin{equation}
\sum_{A\in \mathcal{A}} \sum_{p\in \mbox{ext}(p_{A,a})}\beta_{p,A,a}p=\sum_{A\in \mathcal{A}} \sum_{p\in P_{A,a}} \hat{\beta}_{p,A,a}p,\,\forall a\in X
\end{equation}
\begin{equation}
\sum_{a\in A}\sum_{p\in \mbox{ext}(p_{A,a})}\beta_{p,A,a}p= \sum_{a\in A}\sum_{p\in P_{A,a}} \hat{\beta}_{p,A,a}p,\,\forall A\in \mathcal{A}
\end{equation}
\begin{equation}
\sum_{A\in\mathcal{A}}\sum_{a\in A}\sum_{p\in (\{p_0\}\cap P_{A,a})}\hat{\beta}_{p,A,a}>0
\end{equation}
\end{axiom}
\end{comment}

\begin{axiom}[No Balanced
%\footnote{I refer to this as ``balanced" since, akin to 
%the balancing in the Bondareva-Shapley theorem about the non-emptiness of 
%the core in cooperative game theory, the alternative sets/distributions 
%come out ``on average" to the original test set/distribution. See 
%\cite{bondareva1963some} and \cite{shapley1967balanced}.}
Prior/Outer Point Signals (NBPS)]
\label{NBPS}
For $i\in\{1,...,n\}$, let $A_i\in\mathcal{A}$, weights 
$\beta_i\in\mathbb{R}_+$, and distributions 
$\pi_i\in \Delta(A_i\times \Delta(\Omega))$ such that 
$\mbox{supp}(\pi_i)\subset\mbox{supp}(\pi_{\sigma_{A_i}})$.
There is no sequence of POP signals $\{\tilde{\pi}_i\}_{i=1}^n$
such that 
\begin{equation}
\label{balact}
\sum_{i=1}^n \beta_i p_{A_i,a}\pi_i(a,p_{A_i,a})=\sum_{i=1}^n \beta_i\sum_{p\in P_{A_i,a}} p\tilde{\pi}_i(a,p),\,\forall a\in X
\end{equation}
\begin{equation}
\label{balmenu}
\sum_{a\in A_i}\beta_i p_{A_i,a}\pi_i(a,p_{A_i,a})= \sum_{a\in A_i}\beta_i\sum_{p \in P_{A_i,a}} p\tilde{\pi}_i(a,p),\,\forall i
\end{equation}
\begin{equation}
\label{onprior}
\sum_{i=1}^n\sum_{a\in A_i} \beta_i\tilde{\pi}_i(a,p_0)>0
\end{equation}
\end{axiom}

Axiom \ref{NBPS} can be interpreted as follows. Vector 
$\beta\equiv (\beta_1,...,\beta_n)$ describes a perturbation of the 
distribution of revealed recommendations, where the probabilities of 
recommendations $(a,p)\in \mbox{supp}(\pi_{\sigma_{A_i}})$ decrease 
proportionally to $\beta_i\pi_i(a,p)$, while the probabilities of 
recommendations $a,p$ with $p\in  P_{A_i,a}$ increase proportionally to 
$\beta_i\tilde{\pi}_i(a,p)$. NBPS states that this perturbation cannot satisfy 
the following three properties at the same time:
\begin{enumerate}
    \item The overall state-dependent probability of choosing each 
    action $a$ is preserved on average across menus (\ref{balact}).
    \item This perturbation is Bayes-plausible for each menu 
    (\ref{balmenu}).
    \item The new, perturbed distribution now places positive weight
    on the prior for at least one menu $A$ for which it received probability $0$ according to the revealed posteriors (\ref{onprior}).
\end{enumerate}

The following example illustrates a three-menu case in which NBPS
is violated, even though the support of the revealed posteriors
for all menus is contained in the set of outer points.%
\footnote{We have already seen in Example 1 a case of a violation in 
one menu without support outside the outer points.}

\noindent\textbf{Example 3:} Consider a binary state space, 
with posteriors (with some abuse of notation) given by $p\in[0,1]$, and
a grand set of actions $X=\{a,b,c\}$, where (with further abuse of 
notation) 
\[
v(x,p)=\begin{cases}
    2-3p, & x=a\\
    1, & x=b\\
    3p-1, & x=c
\end{cases}
\]
Consider the menus $A_1=\{a,b,c\}$, $A_2=\{a,b\}$, and $A_3=\{b,c\}$,
with prior $p_0=\frac{1}{2}$. For menu $A_1$, assume that only $a$ and 
$c$ are chosen with positive probability, at revealed posteriors 
$\frac{1}{3}$ and $\frac{2}{3}$, respectively. For menu $A_2$, the 
revealed posteriors for $a$ and $b$ are $0$ and $1$, respectively; 
for menu $A_3$, the revealed posteriors for $b$ and $c$ are $0$ and $1$, 
respectively. 

I now show that given $\sigma_A$, there is a balanced prior/outer-point 
signal. Let $\pi_1=\pi_{\sigma_{A_1}}$; $\pi_2$ place probability 
$\frac{3}{5}$ on $p=0$ and $\frac{2}{5}$ on $p=1$; and $\pi_3$ place 
probability $\frac{2}{5}$ on $p=0$ and $\frac{3}{5}$ on $p=1$, where all 
the same actions are taken at the respective posteriors as before. 
Consider the following alternative distribution $\tilde{\pi}_i$: in menu 
$A_1$, select signals $(a,0)$ and $(c,1)$ with equal probability; in menu 
$A_2$, select signals $(a,\frac{1}{3})$ and $(b,\frac{1}{2})$ with 
probabilities $\frac{3}{5}$ and $\frac{2}{5}$; and in menu $A_3$, select 
signals $(b,\frac{1}{2})$ and $(c,\frac{2}{3})$ with probabilities 
$\frac{2}{5}$ and $\frac{3}{5}$, respectively. Then, set $\beta_1=6$, and $\beta_2=\beta_3=5$.
It is easy to verify algebraically that $\{\tilde{\pi}_i\}_{i=1}^3$ 
constitutes a balanced prior/outer-point signal. That is, for each menu 
$A_i$, Bayes' rule is preserved; the weighted average state-dependent 
probabilities of $a$ and $c$ remain the same (the menus at which the 
respective posteriors are chosen being switched, with the same 
weights), as do the state-dependent probabilities of $b$ (the expected 
posterior given $b$ is $\frac{1}{2}$ both before and after the 
switch of distribution). However, there is now positive weight at
the prior in menus $A_2$ and $A_3$. So, the SDSC violates NBPS. 

\subsection{Main Theorem}

NBPS exploits the following property of optimal persuasion:
there exists some hyperplane such that the revealed posteriors and the
sender's payoff lie on this hyperplane, and for all other posteriors,
the sender's payoff lies below. This follows from the solution of 
persuasion via concavification \citep{Aumann1995,Kamenica2011}.%
\footnote{Relatedly, see also the Lagrangian lemma of 
\cite{caplin2022rationally} for optimal information
choices with posterior-separable information costs.}

\begin{proposition}[\cite{Kamenica2011}, Corollary 2]
\label{KG}
A distribution of posteriors $\pi_A$, with $p_0$ in their convex hull,
is optimal for menu $A$ if and only if there exists 
$\lambda_A\in\mathbb{R}^{\vert \Omega\vert}$,
distinct from the affine hull of $\Delta(\Omega)$ itself,
such that for all $a\in A$ with $a\in\mbox{supp}(\sigma_A)$,  
\begin{equation}
\label{lagrange1}
\lambda_A\cdot p_{A,a}=\sum_{\omega\in\Omega}u(a,\omega)p_{A,a}(\omega)
\end{equation}
while for all $p\in\Delta(\Omega)$,
\begin{equation}
\label{lagrange2}
\lambda_A\cdot p\geq \sum_{\omega\in\Omega}u(a^*_A(p),\omega)p(\omega)
\end{equation}
\end{proposition}

Note that the proposition merely rephrases the concavification result
of \cite{Kamenica2011}, as the concave closure of $\phi_A$,
written $cav(\phi_A)$, is a concave function. As shown in 
\cite{rockafellar1970convex}, Theorem 18.8, any concave function
can be written as an envelope of hyperplanes. Therefore, for any $p_0\in\Delta(\Omega)$, there exists $\lambda_{A,p_0}$ 
such that (\ref{lagrange1}) and (\ref{lagrange2}) hold. Setting 
$\lambda_A=\lambda_{A,p_0}$, and noting that the value of the 
concavification is achieved by $\pi_A$, completes the argument.

We are now ready to present our main theorem. In the following 
definition, the reader is reminded that we rule out the trivial 
case of Bayesian persuasion, in which the sender is indifferent
between all actions.

\begin{defn}
A SDSC dataset \emph{is consistent with nontrivial Bayesian persuasion}
if there exists some sender utility function 
$u:X\times\Omega\rightarrow \mathbb{R}$ such that $\sigma_A$ 
solves (\ref{obj}) via its revealed posteriors as defined by 
(\ref{revpos}) and (\ref{revdist}), and satisfies inequality
(\ref{nontrivial}).
\end{defn}

\begin{theorem}
\label{KGperstest}
A SDSC dataset is consistent with nontrivial Bayesian persuasion
if and only if it satisfies NIAS and NBPS. 
\end{theorem}

The theorem exploits Proposition \ref{KG} in both the direction of 
necessity and sufficiency. In the former direction, consistency
with nontrivial Bayesian persuasion means that for each menu,
there exists some optimal $\lambda_A$ for each menu $A$. 
A Bayes-plausible perturbation in the direction of $\tilde{\pi}-\pi$ 
that preserves the state-dependent probabilities of actions, as described
in (\ref{balact}) and (\ref{balmenu}) in NBPS, should preserve the 
weighted average expected utility across menus. However, recall that the 
hypothesis states that $p_0$ is suboptimal if it is not chosen; thus, 
for the menu $A$ where $p_0$ is now chosen, 
$\lambda_A\cdot p_0>\phi_A(p_0)$.
To compensate for this, there must be some other menu $B$ at 
which the chosen posterior $p$ yields a payoff above the hyperplane, i.e.,
$\lambda_B\cdot p<\phi_B(p)$. But this contradicts the optimality of 
$\pi_B$.

We are now in a position to understand better how the data in Example 1 
was inconsistent with Bayesian persuasion for any preference. What led to 
the existence of a balanced POP signal in Example 1
is that one could split the revealed posterior at $p=0.4$ into posteriors
such that one of them lies on the prior, although the original signal was
informative. As I formalize in the following proposition, in the case 
of a single menu, this is the \emph{only} way that a distribution can be 
inconsistent with Bayesian persuasion.

\begin{proposition}
    \label{onemenu}
    When $\mathcal{A}=\{A\}$, any $\pi_{\sigma_A}$ is inconsistent with 
    nontrivial Bayesian persuasion if and only if
    \begin{enumerate}
        \item $\pi_{\sigma_A}$ is informative;
        \item For some $a^*$ such that $a^*\in\arg\max_{a\in A} \sum v(a,\omega)p_0(\omega)$, 
    there exists $p$ such that both
        \begin{enumerate}
            \item $(a^*,p)\in\mbox{supp}(\pi_{\sigma_A})$ 
            \item $p+\epsilon(p-p_0)\in C_{A,a^*}$ for some $\epsilon>0$.
        \end{enumerate}
    \end{enumerate}
\end{proposition}

It should be noted that the restrictiveness of Proposition \ref{onemenu}
does not hold more generally: as seen in Example 3, when testing for
consistency with more than one menu, there can still be violations of 
NBPS even if all of the posteriors in each menu $A$ are contained in 
$\mbox{out}(\mathcal{C}_A)$.

In the opposite direction of Theorem \ref{KGperstest}, the standard 
tool of revealed preference theory, Farkas' lemma, yields a vector of 
constants from the alternative of NBPS. Indeed, these constants can be 
interpreted directly as the values of $u(a,\omega)$ and $\lambda_A$
needed for optimal persuasion as in Proposition \ref{KG}. Thus one can 
find the relevant payoffs and hyperplanes by construction.

\noindent\textbf{Remark 1:} It is interesting to note that, unlike other 
characterizations of Bayesian persuasion (e.g. \cite{Kamenica2011}), 
the present axiomatization does not require an explicit tiebreaking rule 
of the receiver's choice%
\footnote{Of course, I do assume in favor of 
the sender in the case that the sender is indifferent to information
revelation that would lead to multiple actions being chosen. However, as 
stated in Section 3, this is more in the spirit of standard nondegeneracy 
axioms, in order to rule out indifference between all possible actions.} 
as appears in \cite{Kamenica2011}. Indeed, such a tiebreaking rule would 
be ill-defined in the absence of a direct preference relation over 
outcomes, as the dataset only consists of state-dependent stochastic 
choices.%
\footnote{Recall that we are referring to a breaking of ties of the 
\emph{decision maker's} preferences in favor of the sender, and so the 
flexibility to break indifferences of the \emph{sender} due to incomplete data is not relevant to the tie-breaking rule in question. If the sender is indifferent, then the tie-breaking rule for the purposes of Bayesian persuasion would not matter, while conversely, the issue of tiebreaking is present in the same manner even if more complete preferences indeed determine that the sender's preferences are strict.}
Instead, the tiebreaking rule is implicit from the nonexistence of balanced POP signals. If an action is chosen
with positive probability, then the way to falsify that it is
optimally chosen at its revealed posterior is by such an alternative 
signal. On the other hand, if an action $b$ is never chosen, which could 
be the result of an adverse tiebreaking rule, it could instead be 
rationalized by an alternative sender preference in which $b$ 
simply gives a poor payoff. The revealed preference characterization
therefore renders the condition on receiver behavior of sender-optimal tiebreaking superfluous. 

\noindent\textbf{Remark 2:} An appealing feature of Theorem 
\ref{KGperstest} is that it provides a system of linear inequalities that
can be used to falsify the hypothesis of Bayesian persuasion. From
a computational perspective, this allows for the use of linear programming
algorithms to check the solution, which are well-known to have efficient, 
polynomial-time complexity. At the same time, the solution of Bayesian 
persuasion problems in general is difficult to find except when $\Omega$ 
is small \citep{dughmi2017algorithmic}. This difficulty manifests itself 
in our results as well through the need to enumerate the elements of 
$\mbox{out}(\mathcal{C}_A)$, as this is equivalent to a \emph{vertex 
enumeration} problem. It is an open question whether polynomial time 
algorithms exist in general for such problems; the potential concern is 
that the potential number of vertices can explode in the size of the 
state space. However, for a state space of fixed size, there are 
polynomial-time algorithms that can enumerate $\mbox{out}(\mathcal{C}_A)$,
for instance \cite{avis1991pivoting}.

\subsection{Partial identification}

A common question that arises in the revealed preference literature
is the extent to which the characterization result limits the set 
of testable hypotheses. That is, certain classes of preferences are
indistinguishable in a dataset, and as a result, one can assume without
loss of generality that the preferences are of a certain form.%
\footnote{Take, for instance, the characterization of 
\cite{afriat1967construction}, that it is without loss to assume that 
preferences are convex when considering demand data.}
Below, I describe to what extent the preferences of the sender are
restricted by the identification in Theorem \ref{KGperstest}. 

\begin{proposition}
\label{partialID}
    The set of $\mathcal{U}^*\coloneqq \{(u(a,\cdot))_{a\in X}\}\subset \mathbb{R}^{\vert X\vert \cdot \vert \Omega\vert}$ 
    that is consistent with $\sigma$ satisfies the following:
    \begin{enumerate}[i.]
        \item $\mathcal{U}^*$ is convex.
        \item If $\{u(a,\cdot)\}_{a\in X}\in \mathcal{U}^*$, then all 
        affine transformations of $\{u(a,\cdot)\}_{a\in X}$ are also 
        contained in $\mathcal{U}^*$:
        \[
        \{u(a,\cdot)\}_{a\in X}\in \mathcal{U}^*\implies \{\alpha u(a,\cdot)+\beta\}_{a\in X}\in \mathcal{U}^*,\forall \alpha>0,     \beta\in\mathbb{R}^{\vert\Omega\vert}
        \]
        \item If $b\notin \mbox{supp}(\sigma_A)\,\,\forall A$, and 
        $(u(a,\cdot)_{a\in X\setminus\{b\}},u(b,\cdot))\in\mathcal{U}^*$, 
        then $(u(a,\cdot)_{a\in X\setminus\{b\}},u(b,\cdot)-\kappa_b)\in\mathcal{U}^*$ as well, 
        where $\kappa_b\in \mathbb{R}^{\vert \Omega\vert}_+$.
        \item If $C_{A,b}=\emptyset,\,\forall A\in\mathcal{A}$, then
        $u(b,\cdot)$ is unrestricted.
    \end{enumerate}
\end{proposition}

Notice that there is considerable freedom for the payoffs for actions $b$
that are never chosen from any menu. Once one ensures that 
$u(b,\cdot)\leq \lambda_A$ for all $A\ni b$, one can subtract an arbitrary
amount from $u(b,\cdot)$ and it will still be rationalized for the same
$\sigma_A$. Thus, for instance, one can fix $p$, and assume without loss 
of generality that $u(b,\cdot)=\underline{\lambda}(p)$, where 
$\underline{\lambda}(p)$ defines the hyperplane that supports the 
convexification of $\min \{\lambda_A\cdot p: \, A\in\mathcal{A}\}$ at $p$.

\section{Transparent Motives}

Theorem \ref{KGperstest} shows us that the key axiom, NBPS, captures
necessary and sufficient conditions for the SDSC to be consistent with
the general model of Bayesian persuasion of \cite{Kamenica2011}. However,
in many prominent examples (including the famous judge-prosecutor example),
the preferences of the sender do not depend on the state. In such cases, the
sender is often said to have \emph{transparent motives} (see, e.g., 
\citet{lipnowski2020cheap}). I now explore how the axiom NBPS must be 
strengthened in order to test for the stronger hypothesis of state-independent
sender preferences. 

To illustrate the basic idea of this strengthening, consider the following
example.

\noindent\textbf{Example 4:} Suppose that $A=\{a,b,c\}$, and 
$\Omega=\{\omega_1,\omega_2,\omega_3\}$, with each state equally likely ex
ante. Let the DM's preferences be given by 
\[
v(a,\omega)=0,\,\forall \omega
\]
\[
v(b,\omega)=\begin{cases}
    2, & \omega=\omega_2\\
    -3, & \omega\in\{\omega_1,\omega_3\}
\end{cases}
\]
\[
v(c,\omega)=\begin{cases}
    2, & \omega=\omega_3\\
    -3, & \omega\in\{\omega_1,\omega_2\}
\end{cases}
\]
Hence the DM finds it optimal to choose $b$ if and only if 
$p(\omega_2)\geq 0.6$; similarly, it is optimal to choose $c$ if and only if 
$p(\omega_3)\geq 0.4$.

Now suppose that, according to the SDSC data, the DM chooses deterministically
according to the state: namely, $a$ at state $\omega_1$, $b$ at $\omega_2$,
and $c$ at $\omega_3$. This is clearly consistent with Bayesian persuasion: 
namely, if the sender shared identical preferences with the receiver, this
distribution of revealed posteriors would be optimal.

However, this is \emph{inconsistent} with a sender who has transparent motives.
Notice that, as the distribution of revealed posteriors corresponds to full
information, the sender must benefit from persuasion. If they were to have
transparent motives, then, as their payoffs do not depend on the state, but 
only on the distribution of choices of the receiver, then any distribution of
posteriors that achieves this distribution must be equally good for the sender.
Consider, then, the distribution of posteriors $\hat{\pi}$ which has support
on 
\[
p\coloneqq (p(\omega_1),p(\omega_2),p(\omega_3))\in\{(1/3,2/3,0),(1/3,1/3,1/3),(1/3,0,2/3)\},\] 
where $\pi(p)=1/3,\forall p.$ Notice that, given the receiver's 
preference, this also induces each $x\in A$ to be chosen with probability 
$1/3$. At the same time, there is now a posterior equal to the prior with 
positive probability. If the sender benefits from persuasion, then it 
would be possible to improve his payoff by revealing information relative 
to this posterior. So, the distribution of revealed posteriors, which 
yields the same utility as $\hat{\pi}$ under the hypothesis of transparent 
motives, cannot be optimal. $\square$

The intuition of the previous example indicates how one must modify NBPS in
order to test for transparent motives. Instead of (\ref{balact}), which 
mandates that the \emph{state-dependent} probabilities of each action remain
the same, on average, across menus, one can relax the condition to allow
for the \emph{unconditional} probabilities to remain the same as well. This 
strengthens NBPS by demanding that there not be additional possible 
POP switches.

\begin{defn}[State-Independent No Balanced POP Signals (SI-NBPS)]
\label{sinbps}
    For $i\in\{1,...,n\}$, let $A_i\in\mathcal{A}$, $\beta_i\in\mathbb{R}_+$,
and distributions $\pi_i\in \Delta(A\times \Delta(\Omega))$ such that 
$\mbox{supp}(\pi_i)\subset\mbox{supp}(\pi_{\sigma_{A_i}})$.
There is no sequence of POP signals $\{\tilde{\pi}_i\}_{i=1}^n$
such that 
\begin{equation}
    \label{sibalact}
    \sum_{i=1}^n \beta_i \pi_i(a,p_{A_i,a})=\sum_{i=1}^n \beta_i\sum_{p\in P_{A_i,a}} \tilde{\pi}_i(a,p),\,\forall a\in X
\end{equation}
as well as (\ref{balmenu}) and (\ref{onprior}).
\end{defn}

\begin{theorem}
\label{siperstest}
    A SDSC dataset is consistent with nontrivial Bayesian persuasion by 
    a sender with transparent motives if and only if it satisfies NIAS and
    SI-NBPS.
\end{theorem}

\section{Posterior-Mean Bayesian Persuasion}

In this section, I extend the above results about rationalizing 
the dataset via Bayesian persuasion to the case where the receiver's
and sender's preferences depend only on posterior means. This requires
modification of the definitions and conditions in the problem. Below,
before presenting the main theorem of the section, I modify these as 
needed; all remaining variables are defined analogously to their 
previous definitions.

\subsection{Modification for posterior means}

Consider the state space $[0,1]$, with prior CDF $F_0$ with finite support 
$Z\subset [0,1]$ (without loss of generality including $\{0,1\}$), and
prior mean $z_0$. The set of actions is defined as before. The payoff of 
the DM depends only on the posterior mean, and the DM is an expected 
utility maximizer: 
\[
v: X\times[0,1]\rightarrow\mathbb{R}
\]
This is equivalent to the payoff $v(a,z)$ being affine in $z$. The 
function $u$ for the sender is defined analogously.

Given prior $F_0$, the set of feasible distributions of posterior 
means are those that are \emph{mean-preserving contractions} of the 
prior \citep{rothschild1970increasing,gentzkow2016rothschild}. For 
all $z\in[0,1]$,
\begin{equation}
    \label{MPC}
    I_{F_0,F}(z)\coloneqq\int_0^z [F_0(s)-F(s)]ds\geq0
\end{equation}
The set of feasible CDFs is then 
\[
\mathcal{I}_{F_0}=\{F\in \mathcal{F}: I_{F_0,F}(z)\geq0,\forall z\in [0,1], \mbox{ and } I_{F_0,F}(1)=0\}
\]

Analogously to before, we define the \emph{revealed posterior means} as
\[
z_{A,a}\coloneqq \begin{cases}
    \frac{\sum_{z\in Z} z\sigma_{A}(z)F_0(z)}{\sum_{z\in Z}\sigma_{A}(z)F_0(z)}
\end{cases}
\]
The recommendations, comprising actions $a$ and posterior means $z$, then
are distributed according to
\[
f_{\sigma_{A}}(a,z)=\begin{cases}
    \sigma_A(a), & z=z_{A,a}\\
    0, & \mbox{otherwise}
\end{cases}
\]
The \emph{CDF of revealed recommendations}, $F_{\sigma_A}$, is then given
by
\begin{align*}
    F_{\sigma_A}:\qquad  [0,1] \hspace{5 mm} & \rightarrow \hspace{5 mm} [0,1]\\
    z \hspace{5 mm} & \rightarrow \hspace{5 mm} \sum_{a\in A} f_{\sigma_A}(a,z_{A,a})\mathbf{1}[z_{A,a}\leq z]
\end{align*}
where $\mathbf{1}[\cdot]$ is the indicator function.

\subsection{Preliminary lemmas}

In order for a distribution of posterior means to be optimal, the
following conditions must be satisfied:
\begin{proposition}[\citet{dworczak2019simple}]
\label{DM2019}
Suppose that the sender's utility depends only posterior means of the 
receiver's belief. Then a CDF $F_{\sigma_A}\in\mathcal{I}_{F_0}$ is 
optimal if and
only if there exists 
``price function" $\Lambda_A: [0,1]\rightarrow \mathbb{R}$ such that:
\begin{enumerate}
    \item $\Lambda_A \mbox{ is convex, and affine on all intervals where } I_{F_0,F_{\sigma_A}}(z)>0$
    \item $\Lambda_A(z)\geq \phi_A(z),\forall z\in [0,1]$
    \item $\mbox{supp}(F_A)\subset \{z\in[0,1]:\phi_A(z)=\Lambda_A(z)\}$
\end{enumerate}
\end{proposition}

The test, then, will be to see whether one can rationalize the CDFs
of revealed posterior means via an affine utility function 
$u:X\times[0,1]\rightarrow \mathbb{R}$ and a function $\Lambda_A$ as in
\ref{DM2019}.

In this context, I aim to reduce the set of potential points in the 
support of alternative distributions that one needs to consider for the in 
order to test for optimality. Thus, similar to the results of 
\cite{lipnowski2017simplifying}, I 
construct a finite set of posterior means for each menu that are 
sufficient for an optimal persuasion strategy, regardless of what
the sender's preferences are. To this end, define the set of points
\[
Z_{A,a}\coloneqq \mbox{ext}(C_{A,a})\cup (Z\cap C_{A,a})
\]
Since the set of points for which $a$ is optimal for the DM for menu 
$A$ is simply an interval, the set $Z_{A,a}$ therefore simply consists of
the endpoints of this interval, along with any values $z\in Z$ contained
in the interval. The reason that we need to consider the latter values
as well is that we may worry about running against the information
constraint (\ref{MPC}) at certain values of $z$, meaning that it might 
not be without loss to focus on the endpoints of the interval.

I now claim that the set of such points is sufficient for an optimal 
Bayesian persuasion distribution over posterior means. To demonstrate this,
for a given $z_{A,a}\in C_{A,a}$, define the values $z_1,z_2\in Z_{A,a}$
to be \emph{consecutive} if there does not exist 
$z\in (z_1,z_2)\cup Z_{A,a}$.

\begin{lemma}
    \label{outmsuff}
    It is without loss to consider persuasion strategies with support on 
    $Z_{A,a}$. Moreover, the values of $z\in Z_{A,a}$ for each 
    $a\in A$ can be consecutive.
\end{lemma}

\begin{comment}

******

The following lemma establishes the sufficiency of checking a finite set 
of posterior means $z$ to check where a particular distribution $F_i$
satisfies (\ref{MPC}) with equality at $z$.

\begin{lemma}
    \label{altfinsuff}
    $I_{F_0,F}$ can only switch from $0$ to positive or vice versa at 
    $z\in Z$.
\end{lemma}

\begin{proof}
    Suppose at some $z^*\in [0,1]\setminus Z$, $I_{F_0,F}(z)>0$ for 
    $z\nearrow z^*$, but $I_{F_0,F}(z)=0$. Therefore $F(z^*)>F_0(z^*)$. Let 
    $\hat{z}=\min\{z\in Z:\,z>z^*\}$; then for $z\in (z^*,\hat{z})$,
    $I_{F_0,F}(z)<0$, which is impossible. 

    Conversely, suppose that at $z^*$, $I_{F_0,F}(z)>0$ for $z\searrow z^*$,
    but $I_{F_0,F}(z)=0$. Then $\int_{z^*}^z [F_0(s)-F(s)]ds>0$, for all
    $z$ sufficiently close to $z^*$. But then $z^*\in \mbox{supp}(F_0)=Z$, contradiction.
\end{proof}

Lemma \ref{altfinsuff} establishes that there is a finite set of values of 
$z$ which determine the intervals for which $I_{F_0,F}$ is positive. Thus, 

******
    
\end{comment}

It will turn out that Lemma \ref{outmsuff} implies that one can find a 
function $\Lambda_A$ that is piecewise affine over a finite set of 
intervals, by restricting the optimality check of a candidate solution 
to a finite set of points. Therefore, if we find a finite set of vectors 
$\{\lambda_A^k\}$ that generate the convex, piecewise affine price 
function $\Lambda_A$, it is sufficient for it to satisfy condition (2) 
of Proposition \ref{DM2019} with respect to $z\in\bigcup_{a\in A} Z_{A,a}$
(rather than check for all $z\in[0,1]$) for it to be a price 
function that rationalizes $\sigma_A$. For such a piecewise affine
function $\Lambda_A$, one can write the affine component formed by each
$\lambda_A^k$ as $\lambda_{A,1}^k z+\lambda_{A,0}^k$. Let $K_A$ be the 
number of such $\lambda_A^k$ for a given $A$, $k(z_{A,a})$ be the value 
of $k$ for which $\lambda_A^k$ gives the price at $z_{A,a}$, and 
$\lambda_{A}^0$ be the value of $\lambda_A^k$, for $k\in \{1,...,K_A\}$,
that gives the affine component of the price at $z_0$.

\begin{comment}
It therefore follows that one need only check at $z\in Z$ whether 
$I_{F_{0},F}\geq 0$ in order to verify that $F$ is a feasible CDF. Suppose 
at some $z^*\notin Z$, $I_{F_0,F}(z)=0$, while for $z\searrow z^*$, 
$I_{F_0,F}(z)<0$. Then $F(z^*)>F_0(z^*)$. Let 
$\hat{z}=\min\{z\in Z:\,z>z^*\}$. Since 
$\mbox{supp}(F_0)\cup (z^*,\hat{z})=\emptyset$, we have 
$I_{F_0,F}(\hat{z})<0$.
\end{comment}

\subsection{Testing for posterior-mean Bayesian persuasion}

As before, we also need to rule out the trivial case where the sender
is indifferent between all actions. I therefore continue to assume
that the sender sends an informative signal if and only if he benefits
from persuasion. However, the characterization of a sender benefiting from 
persuasion differs when considering persuasion over posterior means,
due to the optimum being characterized by a convex function rather than
a hyperplane. I present this characterization in the following proposition.

\begin{proposition}
    \label{pmben}
    When sender and receiver preferences depend only on posterior means,
    persuasion is strictly beneficial if and only if, for the optimal CDF 
    of posterior means $F_A$, either $\Lambda_A$ is not affine, or 
    $\phi_A(z_0)<\Lambda_A(z_0)$.
\end{proposition}

In Figure 4, I illustrate how the sender benefits from
persuasion in each of the cases of Proposition \ref{pmben}. In the former
(Figure 4(a)), as $\Lambda_A$ is not affine, the values of 
$z\in\mbox{supp}(F)$ that lie in different affine components from that of 
$z_0$ must lie above the affine component containing $\phi_A(z_0)$. In the 
latter (Figure 4(b)), since $\Lambda_A$ is affine (i.e., defines a 
hyperplane), the optimal persuasion becomes concavification as in the 
model in Section 3. Thus, the sender benefits for the same reason, namely, 
that $\phi_A(z_0)$ lies below the optimal hyperplane.

\begin{figure}[t!]
\centering
\begin{subfigure}[t]{0.4\textwidth}
\begin{center}
\resizebox{60 mm}{45 mm}{
\begin{tikzpicture}
\begin{axis}[
	axis lines=center,
  ymin=0.84,ymax=2.04,
  xmin=-0.02,xmax=1.02,
  xlabel={$z$}, ylabel={$u$},
  ymajorticks=false,
    xtick = {0.5}
]
\addplot[blue, line width = 2,samples=200][domain=0:0.3] {2-10/3*x};
\addplot[blue, line width = 2,samples=200][domain=0.3:0.5] {1+0.1*(x-0.3)};
\addplot[blue, line width = 2,samples=200][domain=0.5:0.7] {1.02-0.1*(x-0.5)};
\addplot[blue, line width = 2,samples=200][domain=0.7:1] {1+10/3*(x-0.7)};
\addplot[olive,line width = 2,dashed,samples=200][domain=0:0.4]{2-2.45*x};
\addplot[olive,line width = 2,dashed,samples=200][domain=0.4:0.6]{1.02};
\addplot[olive,line width = 2,dashed,samples=200][domain=0.6:1]{1.02+2.45*(x-0.6)};
\draw[orange, dashed] (axis cs: 0.5,0.84) -- (axis cs: 0.5,2.04);
\node[fill, olive, circle, inner sep=1.5pt] 	at (axis cs:0,2){};
\node[fill, olive, circle, inner sep=1.5pt] 	at (axis cs:1,2){};
\node[fill, olive, circle, inner sep=1.5pt] 	at (axis cs:0.5,1.02){};
\end{axis}
\end{tikzpicture}
}

\caption{Convex $\Lambda_A$}
\end{center}
\end{subfigure}
\hspace{1 mm}
\begin{subfigure}[t]{0.4\textwidth}
\begin{center}
\resizebox{60 mm}{45 mm}{
\begin{tikzpicture}
		\begin{axis}[
			samples=500,
			% ticks=none,
			ytick = {0},
			yticklabels={$0$},
			xtick={0.5},
			ymajorticks=false,
            ymax=1,
            ymin=-0.1,
			xmin=-0.1,
			xmax=1.1,
			axis on top=false,
			axis x line = middle,
			axis y line = middle,
			axis line style={black},
			ylabel={$u$},
			xlabel={$z$},
		]     
        \addplot[blue,samples=200,line width = 2][domain=0:0.25] {0.25+x};
        \addplot[blue,samples=200,line width = 2][domain=0.25:0.5] {0.5-(x-0.25)};
        \addplot[blue,samples=200,line width = 2][domain=0.5:0.75] {0.25+(x-0.5)};
        \addplot[blue,samples=200,line width = 2][domain=0.75:1] {0.5-(x-0.75)};
        \draw[blue] (axis cs: 0.4, 0.3) node[left] {$\phi_A$};
        \draw[orange, dashed] (axis cs: 0.5,0) -- (axis cs: 0.5,1);
        \node[fill, red, circle, inner sep=1.5pt] 	at (axis cs:0.5,0.25){};
        \addplot[olive, dashed, samples=200, line width = 1][domain=0:1] {0.5};
        \node[fill, olive, circle, inner sep=1.5pt] 	at (axis cs:0.25,0.5){};
        \node[fill, olive, circle, inner sep=1.5pt] 	at (axis cs:0.75,0.5){};
        \draw[olive] (axis cs: 0.9, 0.5) node[above] {$\lambda_A$};
        \end{axis}
\end{tikzpicture}
}

\caption{$\phi_A(z_0)<\Lambda_A(z_0)$}
\end{center}
\end{subfigure}

\caption{Benefiting from persuasion over posterior means}
\end{figure}
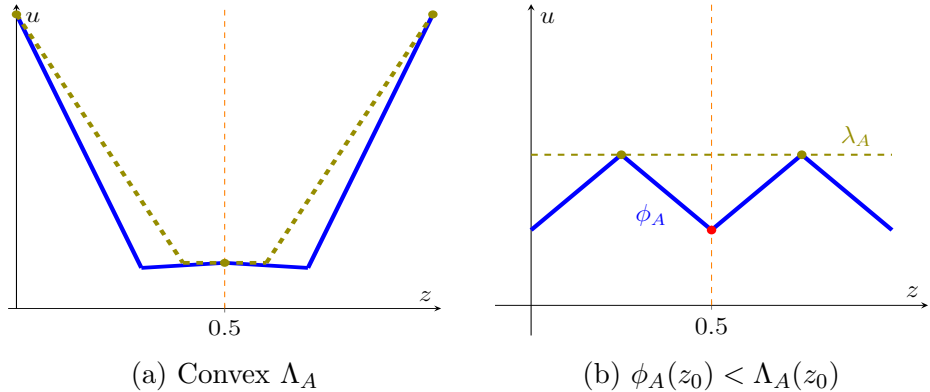

Proposition \ref{pmben} highlights what is needed for a violation of
optimality of the persuasion strategy: one needs to rule out, for a given
distribution, \emph{both} that the sender could be benefiting from a 
nonlinear price function $\Lambda_A$, and that the prior mean lies below the
price function. In addition, one needs to modify NBPS to account for
the differences in the directions one can perturb the distribution 
$F_A$: one can move away from the revealed recommendations only in 
directions that are consistent with the constraint (\ref{MPC}).%
\footnote{See also the discussion in \cite{mensch2026posterior} for 
the challenges in modifying axioms when replacing ``posteriors"
with ``posterior means."}
I encapsulate these in the analogue of NBPS for posterior means.

\begin{axiom}[No Balanced Prior/Outer-Point Signals over Means (NBPS-M)]
    \label{NBPS-M}
    For $i\in\{1,...,n\}$, consider the sequence of menus
$A_i\in\mathcal{A}$, coefficients $\beta_i\in \mathbb{R}_+$, priors 
$F_0^i\in \mathcal{F}$ such that $\mbox{supp}(F_0^i)\subset Z$, and 
distributions of revealed posterior means $F_i\in \mathcal{I}_{F_0^i}$ 
such that:
\begin{enumerate}[(i)]
    \item $\mbox{supp}(f_i)\subset \mbox{supp}(f_{\sigma_{A_i}})$, and
    \item For all $z\in [0,1]$, 
\[
I_{F_0,F_{\sigma_{A_i}}}(z)=0\implies I_{F^i_0,F_i}(z)=0
\]
\end{enumerate} 
Then there do not exist distributions of recommendations $\{g_i\}_{i=1}^n$,
with support over $\{(\bigcup_{a\in A_i}(a,Z_{A_i,a})\cup (z_0\cap C_{A_i,a}))\}_{i=1}^n$,
respectively, such that:
\begin{enumerate}
    \item $G_i\in\mathcal{I}_{F_0^i},\,\forall i$;
    \item $\sum_{i=1}^n z_{A_i,a} \beta_i f_i(a,z_{A_i,a})=\sum_{i=1}^n \sum_{z: a^*_{A_i}(z)=a} z\beta_i g_i(a,z),\,\forall a$
    \item $\sum_{i=1}^n \beta_i f_i(a,z_{A_i,a})=\sum_{i=1}^n \sum_{z: a^*_{A_i}(z)=a} \beta_i g_i(a,z),\,\forall a$
    \item For some $i$, $I_{F^i_0,G_i}(z)>0, \,\forall z\notin\{0,1\}$,
    and $\sum_{a: z_0\in Z_{A_i,a}} g_i(a,z_0)>0$;
\end{enumerate}
\end{axiom}

In words, NBPS-M states that one cannot choose alternative distributions
of posterior means $\{G_i\}$ (a) in a Bayes-plausible way that (b) 
preserves the (weighted average of the) joint distribution of actions and 
states across menus, while (c) for some menu $A_i$ where nontrivial
information was given under $F_i$, placing positive weight on the prior 
mean $z_0$ and making $I_{F_0,G_i}>0$ for all $z\in(0,1)$. Points (a) and 
(b) ensure that the alternative distributions keep the same utility as 
from $\{\sigma_{A_i}\}_{i=1}^n$, as the joint distribution of actions and 
states remains the same, and the sender's preferences are given by 
expected utility. Point (c) is then critical for ensuring against the 
violation of the hypothesis about benefiting from persuasion. $G_i$ cannot be optimal if $\Lambda_A$ is affine, since there is now weight on the prior under $G_i$. 
At the same time, $I_{F_0,G_i}>0$ ensures that $\Lambda_A$ is not affine. 
Since $\{G_i\}_{i=1}^n$ preserves the same payoffs as $\{F_i\}_{i=1}^n$,
and the former is suboptimal, the latter must be as well.

We summarize this intuition in the main theorem of the section.

\begin{theorem}
\label{PMperstest}
    A SDSC dataset is consistent with nontrivial posterior-mean
    Bayesian persuasion if and only if it satisfies NIAS and NBPS-M. 
\end{theorem}

\section{Discussion and Extensions}

\subsection{Comparison with rational inattention}

It has been noted \citep{caplin2013behavioral,caplin2022rationally}
that the same techniques are used to solve both sender-optimal Bayesian 
persuasion and for information choices of rationally inattentive
decision makers with posterior-separable information costs. In each case, one solves for the posterior-%
dependent payoff for the respective agent, and then
 concavifies over these payoffs with respect to the posteriors
to find the optimal value (and, indirectly, the optimal posteriors).
Hence it is natural to ask how the two theories relate to
each other in terms of their testable implications.

\cite{denti2022posterior} provides necessary and sufficient conditions 
for a SDSC dataset to be consistent with posterior separable costs. The
key axiom, ``No Improving Posterior Cycles" (NIPC), states 
(informally) that there is no way to reallocate which 
posterior is chosen at each menu in a Bayes-plausible%
\footnote{\cite{denti2022posterior} phrases the axiom as counterfactually
considering alternative priors for each menu. However, the 
axiom can be rephrased to view these alternative priors merely
as directional changes of the revealed distributions that remain
Bayes-plausible. The latter approach is the one I use in the proofs.}
way while increasing expected utility from the decision on average.
He uses this axiom, via Farkas' lemma, to construct a posterior-%
separable cost of information that rationalizes the data.

By contrast, the key axiom in the present work, NBPS, does 
not make any reference to the average expected payoff from 
the choices of the DM. This is because the sender, when considering 
the choice of the DM, does not care about what payoff 
the DM receives, but only on the distribution of choices that 
she takes. As such, NBPS states that there is no Bayes-plausible
alternative way of getting the same joint distribution of choices
and states across decision problems while generating a distribution
of recommendations that is inferior under the hypothesis.

To highlight that NIAC and NBPS are two distinct axioms, 
neither of which subsumes the other, I present the following
two examples.%
\footnote{One could also ask whether the persuasion model can
be interpreted as one where the \emph{sender} faces an
implicit attention cost through the mediation of the choice
by the receiver. However, given known sender preferences,
this would immediately violate NIAS. To see this, take the 
well-known judge-prosecutor example of \cite{Kamenica2011}. 
The sender always prefers, regardless of
the posterior, that the defendant be convicted. For any possible 
attention cost, then, the defendant would always need to be 
convicted, in contrast to their result that the defendant is 
only somethimes convicted.}
In both cases, the state space is binary ($\Omega=\{\omega_1,\omega_2\}$),
and I abuse notation by writing $v(x,p)$ as in Example 1.

\noindent\textbf{Example 1, Revisited:} We already saw that the
SDSC function in Example 1 violated NBPS. However, this can
easily be justified by a measure of uncertainty 
$H: [0,1]\rightarrow \mathbb{R}$ where
\[
H(p)=\begin{cases}
-2+5p, & p\in[0,0.4]\\
0, & p\in(0.4,0.8]\\
4-5p, & p\in(0.8,1]
\end{cases}
\]
The net utility of the DM's choice minus costs of information is
given in Figure 5.

\begin{figure}[h]
\centering
\resizebox{70mm}{60mm}{
\begin{tikzpicture}
\begin{axis}[
	axis lines=center,
  ymin=-0.02,ymax=1.02,
  xmin=-0.02,xmax=1.02,
  xlabel={$p$}, ylabel={$u$},
  ymajorticks=false,
  xtick = {0.4, 0.5, 0.8}
]
\addplot[blue, line width = 2,samples=200][domain=0:0.2] {0.3-x-2+5*x};
\addplot[blue, line width = 2,samples=200][domain=0.2:0.4] {0.8-2*(x-0.2)-2+5*x};
\addplot[blue, line width = 2,samples=200][domain=0.4:0.6] {0.8-2*(x-0.2)};
\addplot[blue, line width = 2,samples=200][domain=0.6:0.8] {0.4+2*(x-0.6)};
\addplot[blue, line width = 2,samples=200][domain=0.8:1] {0.1+3*(x-0.8)+4-5*x};
\draw[orange, dashed] (axis cs: 0.5,0) -- (axis cs: 0.5,1);
\addplot[olive,line width = 2,dashed,samples=200]{x};
\draw[olive] (axis cs: 0.6, 0.6) node[above] {$\lambda_A$};
\end{axis}
\end{tikzpicture}
}

\caption{Rationalizing $\sigma_A$ via a posterior-separable cost of attention}
\end{figure}
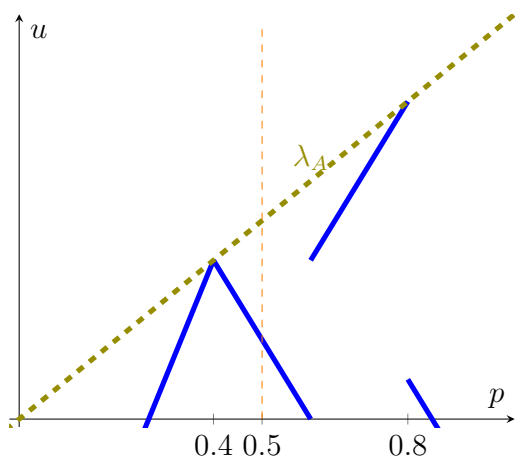

It is easily verified, as the hyperplane indicated by $\lambda_A$ illustrates, 
that $\sigma_A$ is optimal for the posterior-separable cost of 
information defined by $H$, evaluating its expected difference between
prior $p_0=0.5$ and posterior $p$. Thus $(\sigma_A,a)$ satisfies
NIPC. 

I now present an example where the dataset can be rationalized via
Bayesian persuasion by a sender, but not by rational inattention. 

\noindent\textbf{Example 5:} Consider $X={a_1,...,a_5}$, where
\[
v(a_i,p)=\begin{cases}
0.1i-\frac{ip}{3}, & i\in\{1,3\}\\
-0.6+p, & i\in\{2,4\}\\
0, & i=5
\end{cases}
\]
and prior $p_0=0.3$. Consider the menus (Figures 6(a) and 6(b), respectively)
\[
A=\{a_1,a_2,a_5\};\,B=\{a_3,a_4,a_5\}
\]
and suppose that the support for the revealed posteriors are
\[
\mbox{supp}(f_{\sigma_A})=\{(a_1,0);(a_2,1)\};\,\mbox{supp}(f_{\sigma_B})=\{(a_3,0.3)\}
\]

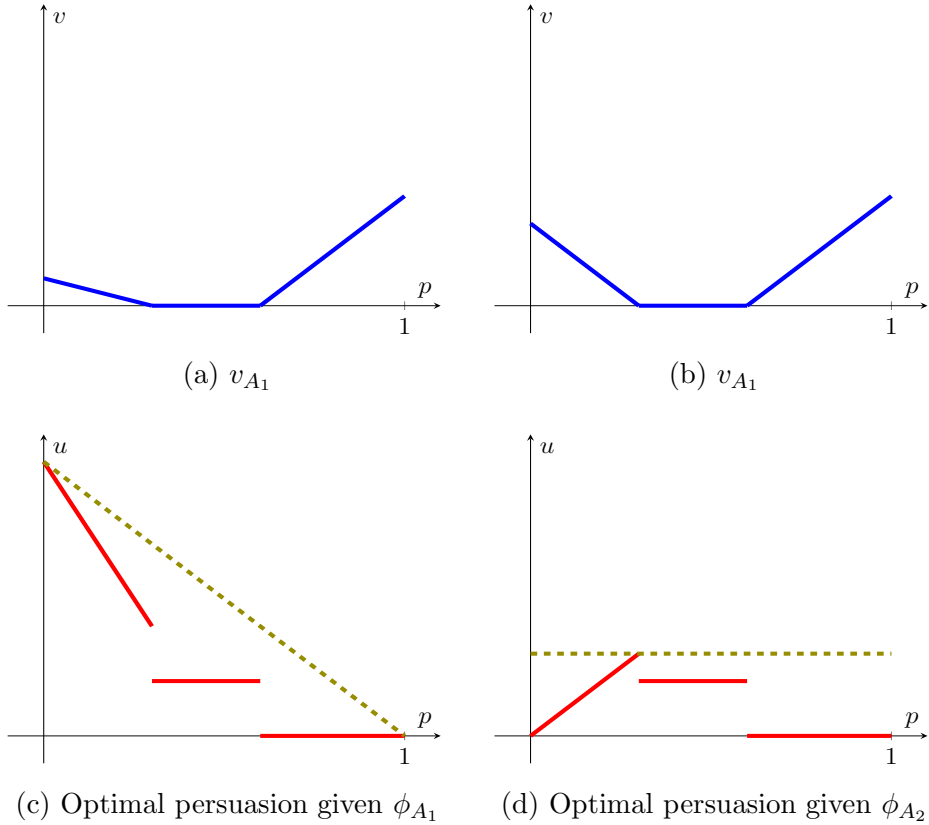
\begin{figure}[t!]
\centering

\begin{subfigure}[t]{0.4\textwidth}
\resizebox{60 mm}{45 mm}{
\begin{tikzpicture}
    \begin{axis}[
        samples=500,
        ytick = {0},
        yticklabels={$0$},
        xtick={0,1},
        ymajorticks=false,
        ymax=1.1,
        ymin=-0.1,
        xmin=-0.1,
        xmax=1.1,
        axis on top=false,
        axis x line = middle,
        axis y line = middle,
        axis line style={black},
        ylabel={$v$},
        xlabel={$p$},
    ]     
    \addplot[blue,samples=200,line width = 2][domain=0:0.3] {0.1-x/3};
    \addplot[blue,samples=200,line width = 2][domain=0.3:0.6] {0};
    \addplot[blue,samples=200,line width = 2][domain=0.6:1] {x-0.6};
    \end{axis}
\end{tikzpicture}
}

\caption{$v_{A_1}$}
\end{subfigure}
\hspace{1 mm}
\begin{subfigure}[t]{0.4\textwidth}
\resizebox{60 mm}{45 mm}{
\begin{tikzpicture}
    \begin{axis}[
        samples=500,
        ytick = {0},
        yticklabels={$0$},
        xtick={0,1},
        ymajorticks=false,
        ymax=1.1,
        ymin=-0.1,
        xmin=-0.1,
        xmax=1.1,
        axis on top=false,
        axis x line = middle,
        axis y line = middle,
        axis line style={black},
        ylabel={$v$},
        xlabel={$p$},
    ]     
    \addplot[blue,samples=200,line width = 2][domain=0:0.3] {0.3-x};
    \addplot[blue,samples=200,line width = 2][domain=0.3:0.6] {0};
    \addplot[blue,samples=200,line width = 2][domain=0.6:1] {x-0.6};
    \end{axis}
\end{tikzpicture}
}

\caption{$v_{A_2}$}
\end{subfigure}

\vspace{5 mm} % Space before the next row

\begin{subfigure}[t]{0.4\textwidth}
\resizebox{60 mm}{45 mm}{
\begin{tikzpicture}
    \begin{axis}[
        samples=500,
        ytick = {0},
        yticklabels={$0$},
        xtick={0,1},
        ymajorticks=false,
        ymax=1.1,
        ymin=-0.1,
        xmin=-0.1,
        xmax=1.1,
        axis on top=false,
        axis x line = middle,
        axis y line = middle,
        axis line style={black},
        ylabel={$u$},
        xlabel={$p$},
    ]     
    \addplot[red,samples=200,line width = 2][domain=0:0.3] {1-2*x};
    \addplot[red,samples=200,line width = 2][domain=0.3:0.6] {0.2};
    \addplot[red,samples=200,line width = 2][domain=0.6:1] {0};
    \addplot[olive,samples=200,dashed,line width = 2][domain=0:1] {1-x};
    \end{axis}
\end{tikzpicture}
}

\caption{Optimal persuasion given $\phi_{A_1}$}
\end{subfigure}
\hspace{1 mm}
\begin{subfigure}[t]{0.4\textwidth}
\resizebox{60 mm}{45 mm}{
\begin{tikzpicture}
\begin{axis}[
        samples=500,
        ytick = {0},
        yticklabels={$0$},
        xtick={0,1},
        ymajorticks=false,
        ymax=1.1,
        ymin=-0.1,
        xmin=-0.1,
        xmax=1.1,
        axis on top=false,
        axis x line = middle,
        axis y line = middle,
        axis line style={black},
        ylabel={$u$},
        xlabel={$p$},
    ]     
    \addplot[red,samples=200,line width = 2][domain=0:0.3] {x};
    \addplot[red,samples=200,line width = 2][domain=0.3:0.6] {0.2};
    \addplot[red,samples=200,line width = 2][domain=0.6:1] {0};
    \addplot[olive,samples=200,dashed,line width = 2][domain=0:1] {0.3};
    \end{axis}
\end{tikzpicture}
}

\caption{Optimal persuasion given $\phi_{A_2}$}
\end{subfigure}

\caption{Rationalizing $\sigma_{A_1},\sigma_{A_2}$ via Bayesian persuasion}
\end{figure}

This is inconsistent with NIAC \citep{caplin2015revealed}, and certainly 
with NIPC \citep{denti2022posterior}, as switching 
$0.3\delta_{1}+0.7\delta_{0}$ with $\delta_{0.3}$ across the two menus, 
where $\delta_p$ is the Dirac measure on posterior $p$ (i.e., switching the 
information choices for the two menus) increases the expected payoff of 
the DM, regardless of what the information cost is. If one is willing to 
acquire a certain amount of information in one menu when the stakes are low, 
then one would not be willing to acquire less information when the stakes 
are higher.

However, this information choice is easily justified by Bayesian 
persuasion. For instance, under the payoffs
\[
u(a_i,p)=\begin{cases}
1-2p, & i=1\\
p, & i=3\\
0, & i\in\{2,4\}\\
0.2, & i=5
\end{cases}
\]
the stated distributions of revealed posteriors are optimal (Figures 6(c)
and 6(d)): since the sender prefers the DM to have information in menu 
$A_1$, but not in $A_2$, one can explain the differences in stochastic 
choices between the two menus, despite the increase in benefit to the DM 
in $A_2$. Thus, this SDSC would satisfy NBPS.

\subsection{Varying the priors}

Changes in the menu that the DM faces may, naturally, also come along
with changes in the prior when she face the menu. One may therefore
wonder how to extend the main result to datasets where the prior
depends on the menu. Indeed, previous work on revealed preference in the 
context of rational inattention \citep{caplin2022rationally, 
denti2022posterior, mensch2025posterior} discuss the properties that
distinguish such theories from those that are restricted to a single
prior. Below, I discuss the implications for both how optimal persuasion 
must look, and how to modify NIAS and NBPS to adapt to variations in prior.

Similarly to \cite{denti2022posterior} and \cite{mensch2026posterior},
the adaptation to variable priors involves a straightforward 
strengthening of NBPS, replacing the fixed prior of Sections 2
and 3 with a variable prior. I index the pairs of menus and priors 
by $i$, allowing for repeats of menus with potentially different 
priors. Thus, for each menu $A_i$, let the respective prior be 
$p_0^i$, which has full support over $\Omega$; the respective
variables are modified analogously as needed to include
dependence on $p_0^i$ in their subscript.

I now present the modifications of the two main axioms.

\begin{axiom}[Uniform No Improving Action Switches (UNIAS)]
\label{UNIAS}
For all $(A_i,p_0^i)$ and $a\in\mbox{supp}(\sigma_{A_i,p_0^i})$,
$b\in A_i$,
\begin{equation}
\sum_{\omega\in\Omega} v(a,\omega)p_{A_i,p_0^i, a}(\omega)\geq \sum_{\omega\in\Omega} v(b,\omega)p_{A_i,p_0^i, a}(\omega)
\end{equation}
\end{axiom}

\begin{axiom}[Uniform No Balanced Prior/Outer-Point Signals (UNBPS)]
\label{UNBPS}
For $i\in\{1,...,n\}$, let $A_i\in\mathcal{A}$, $\beta_i\in\mathbb{R}_+$,
and distributions $\pi_i\in \Delta(A\times \Delta(\Omega))$ such that 
$\mbox{supp}(\pi_i)\subset\mbox{supp}(\pi_{\sigma_{A_i,p_0^i}})$.
There is no sequence of POP signals $\{\tilde{\pi}_i\}_{i=1}^n$
such that%
\footnote{The potential support of $\tilde{\pi}_i$ depends only on 
$p_0^i$ through Bayes plausibility and the changing value 
of $p_0^i\in P_{A_i,p_0^i,a}$, since the posterior cover depends only
on $A_i$. Thus, if $A_i$=$A_j$ but $p_0^i\neq p_0^j$, then
\[
(P_{A_i,p_0^i,a}\cup P_{A_j,p_0^j,a})\setminus (P_{A_i,p_0^i,a}\cap P_{A_j,p_0^j,a})=\{p_0^i,p_0^j\}
\]}
\begin{equation}
\label{ubalact}
\sum_{i=1}^n \beta_i p_{A_i,p_0^i,a}\pi_i(a,p_{A_i,p_0^i,a})=\sum_{i=1}^n \beta_i\sum_{p\in P_{A_i,p_0^i,a}} p\tilde{\pi}_i(a,p),\,\forall a\in X
\end{equation}
\begin{equation}
\label{ubalmenu}
\sum_{a\in A_i}\beta_ip_{A_i,p_0^i,a}\pi_i(a,p_{A_i,p_0^i,a})= \sum_{a\in A_i}\beta_i\sum_{p \in P_{A_i,p_0^i,a}} p\tilde{\pi}_i(a,p),\,\forall i
\end{equation}
\begin{equation}
\label{uonprior}
\sum_{i=1}^n\sum_{a\in A_i} \tilde{\pi}_i(a,p_0^i)>0
\end{equation}
\end{axiom}

\begin{theorem}
\label{UKGperstest}
A SDSC dataset with varying priors is consistent with nontrivial 
Bayesian persuasion if and only if it satisfies UNIAS and UNBPS. 
\end{theorem}

I omit the proof of Theorem \ref{UKGperstest} since, as in the 
related literature on such extensions, it is substantially
identical to that of Theorem \ref{KGperstest}.

As a final remark, I discuss whether the dataset will satisfy an
analogue of ``Locally Invariant Posteriors" (LIP) of 
\cite{caplin2022rationally} as one varies the prior. LIP states that, 
if one changes the prior at menu $A$ from $p_0$ to $\hat{p}_0$, while 
$\hat{p}_0\in \mbox{co}(\bigcup_{a\in \mbox{supp}(\sigma_A)}P_{A,a})$, 
then there exists an optimal signal $\pi_{A,\hat{p}_0}$ such that 
$\mbox{supp}(\pi_{A,\hat{p}_0})\subset \mbox{supp}(\pi_{A,p_0})$.
The answer is partially in the affirmative. Notice that the proof 
of the proposition involves constructing payoffs for each action, 
and an optimal hyperplane for each menu, such that $\phi_A(p_0)$ 
lies strictly below this hyperplane. For $\hat{p}_0$, the same 
hyperplane will remain optimal, and so $\pi_{A,\hat{p}_0}$ (whose 
support lies within that of $\pi_{A,p_0}$) will be optimal. Thus,
the constructed utility functions from datasets that satisfy
UNIAS and UNBPS satisfy LIP. 

However, it may not be the case that the sender continues to 
\emph{strictly} benefit from persuasion if the prior changes, 
as seen in the following example.

\noindent\textbf{Example 6:} Suppose that there are two states 
$\Omega=\{\omega_1,\omega_2\}$, and we indicate the probability 
of $\omega_2$ by $p$ as in previous examples. Let $A=\{a,b,c\}$
be optimal for the DM in the regions 
$\{[0,\frac{1}{3}],[\frac{1}{3},\frac{2}{3}],[\frac{2}{3},1]\}$,
respectively. The payoffs for the sender for each action are $0$
for $b$, and $1$ for $a$ and $c$. 

Starting with $p_0=0.5$, suppose that $\mbox{supp}(\pi_{A,0.5})=\{(a,0);(c,1)\}$. 
Then $\sigma_{A,0.5}$ is optimal, and the sender strictly benefits
from persuasion. This is no longer the case, though, when 
$\hat{p_0}=0.2$: the support remains optimal, but the sender no
longer strictly benefits; thus LIP is not sufficient for the
sender to strictly benefit.

\subsection{Exclusion of other posteriors}

The hypothesis tested in Theorem 1 was that the prior is 
strictly suboptimal if the sender provides nontrivial information.
However, one might be interested in testing the hypothesis
that other posteriors $p\neq p_0$ are also strictly suboptimal 
if not chosen. One can easily accommodate such hypotheses 
by modifying (12) in NBPS to also hold for these other 
posteriors $\{p\}$ as well. Thus, there would be no way to 
keep the same joint distribution of actions and states across
menus while placing positive weight at any such $p$ counterfactually.
Formally, this would alter $P_{A,a}$ to include some additional
points beyond $\mbox{ext}(C_{A,a})$ (and $p_0$, if $p_0\in C_{A,a}$). 
Let us label the set of such points, for a given $(A,a)$, as 
$\bar{P}_{A.a}$. One would then alter inequality (\ref{onprior}) 
to state that the sum of the weights on such points must 
be strictly positive:
\[
\sum_{i=1}^n\sum_{a\in A_i} \sum_{p\in \bar{P}_{A_i,a}} \tilde{\pi}_i(a,p)>0
\]
The exercise for doing this is analogous to that
of testing with respect to different priors in the previous
subsection (as the different priors themselves serve as different
posteriors for which to be tested), and so the proof is omitted
for the same reason.

\subsection{Testing for receiver preferences}

Lastly, the entire analysis in this paper has focused on the case 
where the receiver's preferences are known, but the sender's 
are not. What can be said when this is reversed: the sender's
preferences are known, but the receiver's are not?

Part of what makes the former case easier to solve is that
the optimality of $\{\pi_A\}_{A\in\mathcal{A}}$ can be 
represented as a set of linear inequalities in the (unknown)
sender payoffs $u$ and hyperplanes given by $\lambda_A$. This
observation lends itself to the use of Farkas' lemma to find
the relevant condition. Indeed, such a condition, NBPS, 
provides a rather intuitive rationale as to why it is the 
correct one to guarantee a Bayesian persuasion representation.

When the receiver's preferences are unknown, though, one
must characterize the set of preferences for which the distribution
of posteriors $\pi_A$ remains optimal. Even restricting to the case
where the posteriors $p_{A,a}\in\mbox{supp}(\pi_A)$ satisfy 
$\mbox{ext}(p_{A,a})=\{p_{A,a}\}$ (i.e., the revealed posterior
is at an extreme point of $C_{A,a}$), there are many possible sets 
$C_{A,a}$ with which this observation could be consistent. Indeed, 
there could be different preferences $v,\hat{v}$ for which the 
possible sets $C_{A,a}\cap \hat{C}_{A,a}$ (respectively)
are not identical; in the worst case, they might just intersect 
at $p_{A,a}$. Finding an axiomatization for receiver preferences
thus remains an interesting and challenging open question.

\section{Conclusion}

The results in this paper show how, with a sufficiently detailed 
dataset, one can test the hypothesis that a decision maker is being
persuaded by a (possibly unobserved) sender. The key insight is that,
if one can choose an alternative distribution of posteriors that, on
average preserves the joint distribution of actions and states, while 
allowing the sender to do so without persuasion, then the hypothesis
is falsified. This is because the expected utility for such an 
alternative is preserved, regardless of what the sender's preferences
would be; if he does not need to persuade the receiver to achieve this
payoff, then this cannot be the reason for the distribution. 

The use of SDSC data for testing hypotheses of stochastic choice given
information has been well-established in the laboratory. There remain, 
however, several challenges for applying these tests for data outside 
of the laboratory; as a result, such applications, to date, are extremely 
limited. First, it is very difficult to get the level of detailed data 
that is needed: one must observe (i) the true realized state of the world; 
(ii) a sufficient number of decisions from the same menu from decision 
makers with the same preferences; (iii) the preferences of the decision 
makers; and (iv) the priors. Second, decisions may be noisy, for all sorts 
of reasons, and as a result one must account for this when testing the
theory by allowing for some degree of error. 

One promising direction toward constructing a SDSC dataset outside
of the laboratory would be in the context of media engagement on news
websites. Suppose that a media company has a political slant, and so
while they do not falsify the news, they attempt to frame it in a 
particular way to achieve their agenda. Recent work by, for instance, 
\cite{levy2021social} exploits detailed data on social media engagement 
by individual users. Hypothetically, one could similarly construct a 
dataset for a news source's website, where, for a particular news cycle, 
the company decides which stories to cover, where to place them on the 
page, and how to phrase the titles of the stories. One could then track 
which users click on which links. With some assumptions about the priors 
and/or the similarity of preferences of users, or with many independent 
observations of the same user over time, one can potentially get detailed 
enough data for the testing of the hypothesis of Bayesian persuasion
by the media companies.

An alternative direction to try to overcome these limitations would
be to relax the requirements on the datasets. There has been some
work to consider datasets where, instead of seeing the joint 
distributions of actions and states, the analyst sees only the marginal
distributions of each \citep{rehbeck2023revealed, doval2024revealed}. 
However, these works only provide an equivalent of the NIAS axiom of 
\cite{caplin2015testable}, and do not show how to relax those axioms that 
appear in the literature alongside NIAS. In other words, it provides
a way to test whether the choices are consistent with the theory that
the DM is choosing \emph{given} the information at hand, but does not
provide a way of testing \emph{how} that information is generated. 
Developing a testable implication of this weaker hypothesis is an
interesting open question for further work.

In any case, there will need to be some allowance for adaptation of 
the theory to the particular dataset being tested. For instance, 
as mentioned above, one needs to allow for some noise in the process of 
decision making. Laboratory experiments that use SDSC datasets 
\citep{dean2023experimental, denti2022posterior} have attempted to 
deal with this issue by introducing logit error terms in the 
probabilities of choices. Different adaptations may be needed for 
the particular dataset in question, either due to the nature of the
data, or to test the hypothesis of a particular form of sender 
preferences. The present work provides a foundation on
which to construct some adaptations: namely, by checking whether there
are alternative arrangements of posteriors that preserve the joint 
distribution, while placing individual posteriors in a way that violates
the hypothesis, one can test whether the dataset is consistent with 
Bayesian persuasion.

\bibliographystyle{jpe.bst}
\bibliography{IBP}

\pagebreak{}

\appendix

\begin{center}
\section*{Appendix: Proofs}
\end{center}

\section{Proofs from Section 5}

\subsection{Proof of Theorem \ref{KGperstest}: Necessity}

Suppose that $\{\pi_A\}$ is optimal for 
the sender. That NIAS is satisfied is trivial by the incentive
compatibility for the DM of the recommendations of the 
persuasion strategy $\pi_A$ for each menu, namely that $p_{A,a}\in C_{A,a}$. 

It remains to show that NBPS is satisfied. To do this, 
I show that if NBPS were violated, there would exist some
alternative set of distributions $\{\hat{\pi}_A\}_{A\in\mathcal{A}}$
that is feasible and increases the sender's payoff in at least
one menu $A^*$. 

Suppose that $\pi_{\sigma_A}$ is optimal for each $A$. Let 
$\{\beta_i,\pi_i,\tilde{\pi}_i\}_{i=1}^n$ satisfy (\ref{balact})-%
(\ref{onprior}) for the sequence of menus $\{A_i\}_{i=1}^n$. 
Now define $\hat{\pi}_{A_i}$ as follows. Let 
\[
\delta\coloneqq\max_i \beta_i
\]
Let 
\begin{equation}
\label{hatpi}
\hat{\pi}_{A_i}(a,p)=\pi_{\sigma_{A_i}}(a,p)-\frac{\beta_i\pi_i(a,p)}{\delta}+\frac{\beta_i\tilde{\pi}_i(a,p)}{\delta},\,\forall p
\end{equation}
By construction, $\hat{\pi}_{A_i}(a,p)\geq 0$, 
$\forall p\in P_{A_i,a}\cup\{p_{A_i,a}\}$, and $0$ elsewhere.
Moreover, by equation (\ref{balmenu}), simple algebra yields
\[
\sum_{a\in A_i}\sum_{p\in P_{A_i,a}\cup \{p_{A_i,a}\}}p\hat{\pi}_{A_i}(a,p)=(1-\frac{\beta_i}{\delta})\sum_{a\in A_i}p_{A_i,a}\pi_{\sigma_{A_i}}(a,p_{A_i,a})+ \frac{\beta_i}{\delta}\sum_{a\in A_i}\sum_{p\in P_{A_i,a}}p\tilde{\pi}_{A_i}(a,p)=p_0
\]
and
\[
\sum_{a\in A_i}\sum_{p\in P_{A_i,a}\cup \{p_{A_i,a}\}}\hat{\pi}_{A_i}(a,p)=(1-\frac{\beta_i}{\delta})\sum_{a\in A_i}\pi_{\sigma_{A_i}}(a,p_{A_i,a}) +\frac{\beta_i}{\delta}\sum_{a\in A_i}\sum_{p\in P_{A_i,a}}\tilde{\pi}_{A_i}(a,p)=1
\]
Thus, $\hat{\pi}_{A_i}$ yields a Bayes-plausible distribution of posteriors
given prior $p_0$. The sum of expected utilities across menus
is given by
\[
\sum_{a\in X}\sum_{i=1}^n \sum_{p}\sum_{\omega\in\Omega}u(a,\omega)p(\omega)\hat{\pi}_{A_i}(a,p)=\sum_{a\in X}\sum_{i=1}^n \sum_{p}\sum_{\omega\in\Omega}u(a,\omega)p(\omega)[\pi_{\sigma_{A_i}}-\frac{\beta_i\pi_i}{\delta}+\frac{\beta_i\tilde{\pi}_i}{\delta}](a,p)
\]
\[
=\sum_{a\in X}\sum_{i=1}^n\sum_{p}\sum_{\omega\in\Omega} u(a,\omega)p(\omega)\pi_{\sigma_{A_i}}(a,p)+\sum_{a\in X}\sum_{\omega\in\Omega}u(a,\omega)[-\sum_{i=1}^n\sum_p p\frac{\beta_i\pi_i(a,p)}{\delta}+\sum_{i=1}^n\sum_p p\frac{\beta_i\tilde{\pi}_i(a,p)}{\delta}]
\]
\begin{equation}
\label{equtil}
    =\sum_{a\in X}\sum_{i=1}^n\sum_{p}\sum_{\omega\in\Omega} u(a,\omega)p(\omega)\pi_{\sigma_{A_i}}(a,p)
\end{equation}
However, recall that by the optimality of $\pi_{\sigma_{A_i}}$ for 
each $A_i$, for all $(a,p)\in\mbox{supp}(\pi_{\sigma_{A_i}})$,
\[
\sum_{\omega\in\Omega}u(a,\omega)p(\omega)=\lambda_{A_i}\cdot p,\,\forall a\in\mbox{supp}(\sigma_{A_i}),i\in\{1,...,n\}
\]
and so 
\[
\sum_{a\in X}\sum_{i=1}^n\sum_{p}\sum_{\omega\in\Omega}u(a,\omega)p(\omega)\pi_{\sigma_{A_i}}(a,p)=\sum_{a\in X}\sum_{i=1}^n\sum_{p}(\lambda_{A_i}\cdot p)\pi_{\sigma_{A_i}}(a,p)
\]
By (\ref{onprior}), there exists some $A_i,\,a\in\ A_i$ such that 
$\hat{\pi}_{A_i}(a,p_0)>0$ and for which 
$\lambda_{A_i}\cdot p_0>\phi_{A_i}(p_0)$. 
Let $a$ be an action chosen with positive probability at menu $A$ 
at posterior $p_0$. By equation (\ref{equtil}), there then exists 
$j\in\{1,...,n\}$ and $p\in P_{A_j,a}$ such that 
$\sum_{\omega\in\Omega}u(a,\omega)p(\omega)>\lambda_{A_j}\cdot p$.
But this contradicts the optimality of $\pi_{\sigma_{A_j}}$. $\square$

\subsection{Proof of Theorem \ref{KGperstest}: Sufficiency}

Inequalities (\ref{balact})-(\ref{onprior})
are equivalent to there not existing nonnegative scalars 
\begin{equation}
\label{hatbeta}
    \hat{\beta}_{i,a,p}\coloneqq \beta_i \pi_i(a,p)
\end{equation}
and
\begin{equation}
\label{tildebeta}
    \tilde{\beta}_{i,a,p}\coloneqq \beta_i\tilde{\pi}_i(a,p)
\end{equation}
satisfying these same inequalities with the right-hand sides
of (\ref{hatbeta}) and (\ref{tildebeta}) replaced with the 
left-hand sides. Define the matrix $\mathbf{A}$
with $\vert\mathcal{A}\vert\times\vert \Omega\vert+\vert X\vert \times\vert\Omega\vert$ 
rows (indexed by $i$, corresponding to pairs of actions or menus 
with states) and $\sum_{A\in\mathcal{A}}\sum_{a\in A} [1+\vert P_{A,a}\vert]$
columns (indexed by $j$, corresponding to posteriors $p=p_{A,a}$ 
or $p\in P_{A,a}$). The entries of $\mathbf{A}$ are given by
\[
\mathbf{A}_{i,j}=\begin{cases}
p(\omega), & i=(A,\omega),\,j=p_{A,a}\\
-p(\omega), & i=(A,\omega),\,j=p\in P_{A,a}\\
-p(\omega), & i=(a,\omega),\,j=p_{A,a}\\
p(\omega), & i=(a,\omega),\,j=p\in P_{A,a}\\
0, & \mbox{otherwise}
\end{cases}
\]
Simultaneously, define the vector $\mathbf{b}$ of length 
$\sum_{A\in\mathcal{A}}\sum_{a\in A} [1+\vert P_{A,a}\vert]$ by 
\begin{equation}
\label{farkas_b}
    \mathbf{b}_i=\begin{cases}
    -1, & p=p_0,\,p\notin\mbox{supp}(\pi_A)\\
    0, & \mbox{otherwise}
\end{cases}
\end{equation}
One can express the failure of conditions (\ref{balact})-(\ref{onprior})
by the nonexistence of vector 
$\begin{pmatrix} \hat{\beta} \\ \tilde{\beta}\end{pmatrix}\geq 0$ that 
satisfies
\[
\mathbf{A}\begin{pmatrix} \hat{\beta} \\ \tilde{\beta} \end{pmatrix}= 0
\]
while
\[
\begin{pmatrix} \hat{\beta} \\ \tilde{\beta}\end{pmatrix}\cdot \mathbf{b}<0
\]
By Farkas' lemma (\cite{aliprantis2006infinite}, Corollary 5.85), there
exists real vector $\begin{pmatrix} \lambda \\ u\end{pmatrix}$ of 
length $\vert\mathcal{A}\vert\times\vert \Omega\vert+\vert X\vert \times\vert\Omega\vert$
such that
\[
\mathbf{A}^T \begin{pmatrix} \lambda \\ u\end{pmatrix}\leq \mathbf{b}
\]
Thus, for each $p\in \mbox{supp}(\pi_{\sigma_{A}})\cup P_{A,a}$, one has the inequalities
\begin{equation}
\label{farkopt}
    \sum_{\omega\in\Omega} u(a,\omega)p_{A,a}(\omega)\geq \lambda_{A}\cdot p_{A,a},\,\forall A,a\in \mbox{supp}(\sigma_A)
\end{equation}
\begin{equation}
\label{farksubopt}
    \sum_{\omega\in\Omega} u(a,\omega)p(\omega)\leq \lambda_{A}\cdot p,\,\forall p\in P_{A,a},\forall A,a
\end{equation}
\[
    \sum_{\omega\in\Omega} u(a,\omega)p(\omega)\leq \lambda_{A}\cdot p-1,\,p=p_0\in P_{A,a},\forall A,a
\]

As a result, for all $p\in\mbox{supp}(\pi_{\sigma_A})$, 
$\sum_{\omega\in\Omega}u(a,\omega)p(\omega)=\lambda_A\cdot p$,
i.e. the posteriors $p$ lie on the optimal hyperplane for menu $A$. 
Moreover, the payoffs at all other posteriors 
$p\in \mbox{out}(\mathcal{C}_A)$ lie weakly below the hyperplane, with 
this being strict at $p_0$: 
\begin{equation}
    \label{farkprior}
    \sum_{\omega\in\Omega} u(a,\omega)p(\omega)< \lambda_{A}\cdot p,\,p=p_0\in P_{A,a},\forall A,a
\end{equation}
Lastly, by Corollary \ref{outsuff}, one can always find an optimal 
persuasion strategy with support on $\{(a,p): a\in A, p\in\mbox{ext}
(C_{A,a})\}$; therefore, as $\pi_{\sigma_A}$ is weakly preferred to all 
signals with this support, it is optimal. $\square$

\subsection{Proof of Proposition \ref{onemenu}}

    Suppose that $\pi_{\sigma_A}$ satisfies conditions 1 and 2 of the 
    proposition. Then for some $(a^*,p_{A,a^*})\in\mbox{supp}(\pi_{\sigma_A})$
    such that $p_0\in C_{A,a^*}$, there exists vector 
    $\beta\in \Delta(P_{A,a^*})$ such that 
    $\sum_{p\in P_{A,a^*}}p\beta(p)=p_{A,a^*}$ and $\beta(p_0)>0$. 
    But then there is a balanced POP signal that replaces signal 
    $(a^*,p)$ with the signals $\{a^*,p\}_{p\in P_{A,a^*}}$ with 
    respective probabilities $\beta\pi_{\sigma_A}(a^*,p_{A,a^*})$. 
    By Theorem \ref{KGperstest}, $\pi_{\sigma_A}$ is not 
    rationalizable by nontrivial Bayesian persuasion. 

    Conversely, suppose that $\pi_{\sigma_A}$ does not satisfy 
    conditions 1 and 2 of the proposition. If $\pi_{\sigma_A}$ is 
    uninformative, then the signal is rationalized by setting 
    \[
    u(a,\omega)=\begin{cases}
        1, & (a,p_0)\in\mbox{supp}(\pi_{\sigma_A}),\,\forall \omega\\
        0, & \mbox{otherwise}
    \end{cases}
    \] 
    Otherwise, by hypothesis, $\pi_{\sigma_A}$ is informative, and
    $p_{A,a}\neq p_0,\,\forall a\in A$. For all 
    $a\notin\arg\max \sum v(a,\omega)p_0(\omega)$, set $u(a,\omega)=0$. 
    For all $a\in \arg\max \sum v(a,\omega)p_0(\omega)$, let
    \begin{multline*}
        \mbox{ext}(p_{A,a})\coloneqq \{p\in \mbox{ext}(C_{A,a}:\,\exists \beta\in\Delta(\mbox{ext}(C_{A,a}))\mbox{ s.t. } \beta(p)>0\\
        \mbox{and } \sum_{p\in \mbox{ext}(C_{A,a})}p\beta(p)=p_{A,a}\}
    \end{multline*}
    That is, $\mbox{ext}(p_{A,a})$ is the set of extreme points of 
    $C_{A,a}$ that can be used, with positive probability, to make 
    convex combinations equal to $p_{A,a}$. By condition 2(b), 
    $p_0\notin co(\mbox{ext}(p_{A,a}))$.%
    \footnote{Recall that any distribution $\pi_{\sigma_A}$ that is
    both informative and has support on the prior always falsifies the
    hypothesis, and therefore this case is omitted.}
    Therefore, $co(\mbox{ext}(p_{A,a}))\cap ri(C_{A,a})=\emptyset$. 
    Since $C_{A,a}$ is a convex polytope, by \cite{rockafellar1970convex}, 
    Theorem 11.6, there exists some affine function 
    $\upsilon_a:\Delta(\Omega)\rightarrow\mathbb{R}$ that achieves its 
    maximum in $C_{A,a}$ at exactly the points $p\in \mbox{ext}(p_{A,a})$, 
    which we can without loss set equal to $0$ at $p_{A,a}$. Setting 
    $\sum u(a,\omega)p(\omega)=\upsilon_a(p)$, we have  
    $\upsilon_a(p)< 0$, for all $p\notin\mbox{ext}(p_{A,a})$.

    By construction, then, for all $p\in \Delta(\Omega)$, 
    $\sum u(a^*_A(p),\omega)p(\omega)\leq 0$, while for 
    $(a,p)\in\mbox{supp}(\pi_{\sigma_A})$, $\sum u(a,\omega)p(\omega)=0$.
    Therefore, $\pi_{\sigma_A}$ is an optimal signal for the sender.

\subsection{Proof of Proposition \ref{partialID}}

\begin{enumerate}[i.]
    \item Immediate from Farkas' lemma defining a set of linear 
    inequalities (\ref{farkopt})-(\ref{farkprior}), and 
    \cite{rockafellar1970convex}, Corollary 2.1.1.
    \item As all inequalities (\ref{farkopt})-(\ref{farkprior})  are 
    affine in $(u,\lambda)$, they are preserved by affine 
    transformation. Therefore, if one applies affine transformation
    $\tilde{u}=\alpha u+\beta$, then applying the transformation 
    $\tilde{\lambda}_A=\alpha\lambda_A+\beta$ also rationalizes the 
    dataset.
    \item For all $A\in\mathcal{A}$ such that $b\in A$, $b$ is suboptimal 
    as a recommendation for the receiver if and only if 
    \[\sum_{\omega} u(b,\omega)p(\omega)\leq \lambda_A\cdot p,\forall p\in \]
    These inequalities are preserved by subtracting $\kappa_b$ from the 
    left-hand side.
    \item Immediate since $u(b,\cdot)$ does not appear anywhere in (\ref{farkopt})-(\ref{farkprior}). 
\end{enumerate}

\section{Proofs from Section 6}

The proof closely follows that of Theorem \ref{KGperstest}. In the 
direction of necessity, NIAS is satisfied as before by $\pi_{\sigma_A}$.
At the same time, if SINBPS were not to hold, let 
$\{\beta_i,\pi_i,\tilde{\pi}_i\}_{i=1}^n$ satisfy (\ref{balmenu}), 
(\ref{onprior}), and construct $\hat{\pi}_i$ as before in equation 
(\ref{hatpi}). The sum of expected utilities across menus under 
$\hat{\pi}$ is 
\[
\sum_{a\in X}\sum_{i=1}^n \sum_{p}u(a)\hat{\pi}_{A_i}(a,p)=\sum_{a\in X}\sum_{i=1}^n \sum_{p}u(a)[\pi_{\sigma_{A_i}}-\frac{\beta_i\pi_i}{\delta}+\frac{\beta_i\tilde{\pi}_i}{\delta}](a,p)
\]
\[
\sum_{a\in X}\sum_{i=1}^n\sum_{p} u(a)\pi_{\sigma_{A_i}}(a,p)+\sum_{a\in X}u(a)[-\sum_{i=1}^n\sum_p \frac{\beta_i\pi_i(a,p)}{\delta}+\sum_{i=1}^n\sum_p \frac{\beta_i\tilde{\pi}_i(a,p)}{\delta}]
\]
\[
=\sum_{a\in X}\sum_{i=1}^n\sum_{p}u(a)\pi_{\sigma_{A_i}}(a,p)
\]
However, recall that by the optimality of $\pi_{\sigma_{A_i}}$ for 
each $A_i$, for all $(a,p)\in\mbox{supp}(\pi_{\sigma_{A_i}})$,
\[
u(a)=\lambda_{A_i}\cdot p,\,\forall a\in\mbox{supp}(\sigma_{A_i}),i\in\{1,...,n\}
\]
and so 
\[
\sum_{a\in X}\sum_{i=1}^n\sum_{p}u(a)\pi_{\sigma_{A_i}}(a,p)=\sum_{a\in X}\sum_{i=1}^n\sum_{p}(\lambda_{A_i}\cdot p)\pi_{\sigma_{A_i}}(a,p)
\]
Thus as in the proof of Theorem \ref{KGperstest}, there exists
$j\in\{1,...,n\}$ and $p\in P_{A_j,a}$ such that 
$\sum_{\omega\in\Omega}u(a,\omega)p(\omega)>\lambda_{A_j}\cdot p$,
contradicting the optimality of $\pi_{\sigma_{A_j}}$.

For sufficiency, define the matrix $\mathbf{A}$ with 
$\vert\mathcal{A}\vert\times\vert \Omega\vert+\vert X\vert $ 
rows (indexed by $i$, corresponding to actions, or pairs of menus 
with states) $\sum_{A\in\mathcal{A}}\sum_{a\in A} [1+\vert P_{A,a}\vert]$
columns (indexed by $j$, corresponding to posteriors $p=p_{A,a}$ 
or $p\in P_{A,a}$). The entries of $\mathbf{A}$ are given by
\[
\mathbf{A}_{i,j}=\begin{cases}
p(\omega), & i=(A,\omega),\,j=p_{A,a}\\
-p(\omega), & i=(A,\omega),\,j=p\in P_{A,a}\\
-1, & i=a,\,j=p_{A,a}\\
1, & i=a,\,j=p\in P_{A,a}\\
0, & \mbox{otherwise}
\end{cases}
\]
Simultaneously, define the vector $\mathbf{b}$ as in (\ref{farkas_b}).
By Farkas' lemma (\cite{aliprantis2006infinite}, Corollary 5.85), there
exists real vector $\begin{pmatrix} \lambda \\ u\end{pmatrix}$ of 
length $\vert\mathcal{A}\vert\times\vert \Omega\vert+\vert X\vert$
such that
\[
\mathbf{A}^T \begin{pmatrix} \lambda \\ u\end{pmatrix}\leq \mathbf{b}
\]
Thus, for each $p\in \mbox{supp}(\pi_{\sigma_{A}})\cup P_{A,a}$, one has 
the inequalities
\[
u(a)\geq \lambda_{A}\cdot p_{A,a},\,\forall A,a\in \mbox{supp}(\sigma_A)
\]
\[
u(a)\leq \lambda_{A}\cdot p,\,\forall p\in P_{A,a},\forall A,a
\]
\[
u(a)\leq \lambda_{A}\cdot p-1,\,p=p_0\in P_{A,a},\forall A,a
\]
As in Theorem \ref{KGperstest}, this implies that $\pi_{\sigma_A}$ is 
optimal for the sender with payoff $u$, and such that $p_0$ is strictly 
suboptimal whenever $\pi_{\sigma_A}$ reveals information to the DM.

\section{Proofs from Section 7}

\subsection{Proof of Lemma \ref{outmsuff}}

Let $F_{A}^*$ be an optimal CDF of posteriors with the fewest 
actions $a$ that have recommendations $(a,z)$ with $z\notin Z_{A,a}$.
If $z_{A,a}\in Z_{A,a}$ for all $a$, then we are done. Otherwise,
suppose that $z_{A,a}\notin Z_{A,a}$. As shown in Section 3, it
is without loss to assume that there is a single such signal realization
$(a,z_{A,a})$ in the support of $f_A^*$. Since $Z_{A,a}$ is finite, 
there exist consecutive $z_1,z_2\in Z_{A,a}$ such that 
$z\in(z_1,z_2)$; let $\alpha=\frac{z-z_1}{z_2-z_1}$. Consider the 
distribution $\hat{F}_A$ defined as follows. For each $a\in A$,
define $\hat{f}_A(a,\cdot)$ by
\[
\hat{f}_A(a,z)=\begin{cases}
    \alpha f_{\sigma_A}(a,z_{A,a}), & z=z_2\\
    (1-\alpha) f_{\sigma_A}(a,z_{A,a}), & z=z_1\\
    0, & \mbox{otherwise}
\end{cases}
\]
Meanwhile, for $b\neq a$, set $\hat{f}_A(b,z)=f_A^*(a,z)$. Thus
$\hat{f}_A$ is a valid probability over $A\times \Delta(\Omega)$.
We now check that $\hat{F}_A$ is a valid CDF given $F_0$.
Clearly, $I_{F_0.\hat{F}_A}(z)=I_{F_0,F_A}(z)$ for all $z\notin (z_1,z_2)$. 
Suppose now that for some $z\in (z_1,z_2),$ $I_{F_0.\hat{F}_A}(z)<0$.
By construction, $F_0(\hat{z})=F_0(z_1),\,\forall \hat{z}\in [z_1,z_2)$.
Therefore, there exists $z^*\in[z_1,z_2)$ such that 
$\hat{F}_A(z^*)>F_0(z^*)$. It follows that 
$I_{F_0,\hat{F}_A}(\hat{z})<0,\,\forall \hat{z}\in(z^*,z_2]$. But this 
contradicts the fact that $I_{F_0,\hat{F}_A}(z_2)=I_{F_0,F_A}(z_2)$. 
Thus $\hat{F}_A$ is a feasible mean-preserving contraction of $F_0$
such that $E_F[u(a,z)]=E_{\hat{F}}[u(a,z)]$. 

\subsection{Proof of Proposition \ref{pmben}}

    Suppose that $\Lambda_A$ is affine and $\phi_A(z_0)=\Lambda_A(z_0)$.
    Then there is a mean-preserving contraction that replaces $F_A$ with 
    $\delta_{z_0}$, for which
    \[
    \sum_{z\in\mbox{supp}(F_A)} \phi_A(z)f_A(a^*(z),z)=\sum_{z\in\mbox{supp}(F_A)} \Lambda_A(z)f_A(a^*(z),z)=\Lambda_A(z_0)
    \]
    Hence, if $\phi_A(z_0)=\Lambda_A(z_0)$, one achieves the same utility 
    from no information, and so persuasion is not strictly beneficial.

    Conversely, suppose that $\Lambda_A$ is not affine. Then
    \begin{align*}
        \sum_{z\in\mbox{supp}(F_A)} \phi_A(z)f_A(a^*(z),z) &=\sum_{z\in\mbox{supp}(F_A)} \Lambda_A(z)f_A(a^*(z),z)\\
        &>\Lambda_A(z_0)\\
        &\geq \phi_A(z_0)
    \end{align*}
    and so the sender strictly benefits from persuasion. Alternatively, if 
    $\phi_A(z_0)<\Lambda_A(z_0)$, then 
    \begin{align*}
        \sum_{z\in\mbox{supp}(F_A)} \phi_A(z)f_A(a^*(z),z) &=\sum_{z\in\mbox{supp}(F_A)} \Lambda_A(z)f_A(a^*(z),z)\\
        &\geq \Lambda_A(z_0)\\
        >\phi_A(z_0)
    \end{align*}
    In either case, $\delta_{z_0}$ is strictly worse for the sender than 
    $F_A$.

\subsection{Proof of Theorem \ref{PMperstest}}

    Suppose that the dataset is the solution to some Bayesian persuasion
    problem when the sender's payoffs are given by $u$. NIAS is trivially
    satisfied, and so it remains to check NBPS-M. Suppose that it is not 
    satisfied for some $\{F_i\},\{G_i\}$. Fix $\epsilon>0$, and for 
    each $i$, define the functions
    \[
    \hat{F}^\epsilon_i=F_{\sigma_{A_i}}-\epsilon\beta_i[F_i-G_i]
    \]
    
    We must check that $\hat{F}_i$ is a feasible distribution
    for sufficiently small $\epsilon$. The following lemma will 
    be useful.

    \begin{lemma}[\cite{mensch2026posterior}, Lemma 5]
        \label{MM2025}
        Suppose that for priors $F_0^i\in \mathcal{F}$ and distributions
        $F_i\in \mathcal{I}_{F_0^i}$, 
        \begin{enumerate}[(i)]
            \item $\mbox{supp}(F_i)\subset \mbox{supp}(F_{\sigma_{A_i}})$,
            and
            \item for all $z\in Z$,
            \begin{equation}
                \label{finsuff}
                I_{F_0,F_{\sigma_{A_i}}}(z)=0\implies I_{F_0^i,F_i}(z)=0
            \end{equation}
            Then (\ref{finsuff}) also holds for all $z\in [0,1]$.
        \end{enumerate}
    \end{lemma}
    
    Returning to the proof, we have that the value of 
    $I_{F_0,\hat{F}_i^\epsilon}(z)$ for $\hat{F}^\epsilon_i$ is
    \[
        \int_0^z [F_0(s)-\hat{F}^\epsilon_i(s)] = \int_0^z [(F_0(s)-F_{\sigma_{A_i}}(s))-\epsilon(F_0^i(s)-F_i(s))+\epsilon(F_0^i(s)-G_i(s))]ds
    \]
    For $z\in Z$, 
    \[
    I_{F_0,F_{\sigma_{A_i}}}(z)=0\iff I_{F_0^i,F_i}(z)=0
    \]
    For $z$ such that $\int_0^z [(F_0(s)-F_{\sigma_{A_i}}(s)]ds=0$,
    we have that $\int_0^z [(F_0^i(s)-F_i(s)]=0$ as well by Lemma 
    \ref{MM2025}. Thus, (\ref{MPC}) is satisfied for $\hat{F}_i(z)$ since 
    $G_i\in \mathcal{I}_{F_0^i}$. We next check for violations of 
    (\ref{MPC}) at values of $z$ such that 
    $\int_0^z [(F_0(s)-F_{\sigma_{A_i}}(s)]>0$. Since $f_0^i(a,z)$ and 
    $f_i(a,z)$ are finite for all $z$, for sufficiently small $\epsilon$, 
    one has
    \[
    \int_0^z [(F_0(s)-F_{\sigma_{A_i}}(s)]ds\geq \epsilon\int_0^z [F_0^i(s)-F_i(s)]ds
    \]
    Therefore, $\hat{F}^\epsilon_i\in \mathcal{I}_{F_0}$. 

    By hypothesis due to the violation of NBPS-M, there exists some
    menu $A_j$ such that $F_{\sigma_{A_j}}\neq \delta_{z_0}$, 
    $I_{F_0,\hat{F}^\epsilon_j}(z)>0$ for all $z\in(0,1)$, and 
    $\sum_{a\in A_j}\hat{f}^\epsilon_j(a,z_0)>0$. As a result, since 
    $F_{\sigma_{A_j}}$ is optimal, $\hat{F}^\epsilon_j$ is 
    suboptimal by Proposition \ref{pmben}. By conditions 2 and 3 
    of NBPS-M, 
    \[
    \sum_{i=1}^n \int_0^1 u(a^*_{A_i}(z),z)dF_{\sigma_{A_i}}(z)=\sum_{i=1}^n \int_0^1 u(a^*_{A_i}(z),z)d\hat{F}^\epsilon_i(z)
    \]
    Therefore, there must be some menu $A_k$ such that 
    \[
    \int_0^1 u(a^*_{A_k}(z),z)dF_{\sigma_{A_k}}(z)<\int_0^1 u(a^*_{A_k}(z),z)d\hat{F}^\epsilon_k(z)
    \]
    which contradicts the optimality of $F_{\sigma_{A_k}}$.
    
    Conversely, suppose that the dataset satisfies NIAS and NBPS-M. 
    For each $A$, let $Z^*_A=\{z\in Z: I_{F_0,F_{\sigma_A}}(z)=0\}$.
    We rewrite NBPS to define a convex cone of vectors using the 
    following lemma.
    
    \begin{lemma}
    \label{pmcone}
        NBPS-M holds if and only if there do not exist nonnegative 
        scalars $\hat{\beta}_{i,a,z_{A_i,a}}$, $\tilde{\beta}_{i,a,z}$, 
        $y_{A_i}$ such that:
    \begin{equation}
        \label{farkm1}
        \sum_{a\in A_i}\hat{\beta}_{i,a,z_{A_i,a}}=\sum_{a\in A_i}\sum_{z\in Z_{A_i,a}}\tilde{\beta}_{i,a,z},\,\forall i
    \end{equation}
    \begin{equation}
        \label{farkm2}
        \sum_{a\in A_i}\sum_{z_{A_i,a}<z^*} [z^*-z_{A_i,a}]\hat{\beta}_{i,a,z_{A_i,a}}\geq \sum_{a\in A_i}\sum_{z\in Z_{A_i,a}, z<z^*} [z^*-z]\tilde{\beta}_{i,a,z}+y_{A_i},\,\forall i,z^*\in Z^*_{A_i}\setminus\{0,1\}
    \end{equation}
    \begin{equation}
        \label{farkm3}
        \sum_{a\in A_i}\sum_{z_{A_i,a}<1} [1-z_{A_i,a}]\hat{\beta}_{i,a,z_{A_i,a}}= \sum_{a\in A_i}\sum_{z\in Z_{A_i,a},z<1} [1-z]\tilde{\beta}_{i,a,z},\,\forall i
    \end{equation}
    \begin{equation}
        \label{farkm4}
        \sum_{i=1}^n z_{A_i,a} \hat{\beta}_{i,a,z_{A_i,a}}=\sum_{i=1}^n \sum_{z\in Z_{A_i,a}}z \tilde{\beta}_{i,a,z},\,\forall a\in X
    \end{equation}
    \begin{equation}
        \label{farkm5}
        \sum_{i=1}^n \hat{\beta}_{i,a,z_{A_i,a}}=\sum_{i=1}^n \sum_{z\in Z_{A_i,a}} \tilde{\beta}_{i,a,z},\,\forall a\in X
    \end{equation}
    \begin{equation}
        \label{farkm6}
        \tilde{\beta}_{i,a,z_0}\geq y_{A_i},\,\forall i
    \end{equation}
    while
    \begin{equation}
        \label{farkm7}
        -\sum_{i=1}^n y_A<0
    \end{equation}
    \end{lemma}

    \begin{proof}
        Let 
    \begin{equation}
        \label{hatbetam}
        \hat{\beta}_{i,a,z}=\beta_i f_i(a,z)
    \end{equation}
    and
    \begin{equation}
        \label{tildebetam}
        \tilde{\beta}_{i,a,z}=\beta_i g_i(a,z)
    \end{equation}

    In the ``only if" direction, equation (\ref{farkm1}) holds due 
    to both $F_i$ and $G_i$ being CDFs; in the ``if" direction, this
    constructs $F_i$ and $G_i$ as CDFs. Equations (\ref{farkm4})-(\ref{farkm7}) 
    are simple rewritings of conditions (1)-(4) of NBPS-M. Thus, it
    remains to show that the existence of some prior
    $F_0^i$ (as in NBPS-M) is equivalent to the conditions 
    (\ref{farkm2}) and (\ref{farkm3}).

    For all $z^*\in Z_{A_i}^*$, $I_{F_0^i,F_i}(z^*)=0$, and therefore the 
    left-hand side of (\ref{farkm2}) equals $\int_0^{z^*}F_0^i(z)dz$;
    since $G_i\in\mathcal{I}_{F_0^i}$, (\ref{farkm2}) is satisfied
    for some $y_A\geq 0$. (\ref{farkm3}) is satisfied for the same 
    reason. 

    Conversely, suppose that
    $\{\tilde{\beta}_{i,a,z_{A_i,a}},\hat{\beta}_{i,a,z},y_{A_i}\}$ as in
    the hypothesis of the lemma. Define $F_0^i$ such that 
    $\mbox{supp}(F_0^i)=Z_{A_i}^*$ and
    \[
    \int_{0}^{z^*} F_0^i(z)dz=\frac{1}{\beta_i}\sum_{a\in A_i,\,z_{A_i,a}<z^*} [z^*-z_{A_i,a}]\hat{\beta}_{i,a,z_{A_i,a}}
    \]
    where $\beta_i$ is as in (\ref{hatbetam}). Then 
    $G_i\in \mathcal{I}_{F_0^i}$: by (\ref{farkm2}) and (\ref{farkm3}),
    $I_{F_0^i,G_i}(z^*)\geq 0$ for $z^*\in Z_{A_i}^*$. If there is some
    other $z\in [0,1]$ where $I_{F_0^i,G_i}(z)< 0$, then at such $z$,
    $G_i(z)>F_0^i(z)$. Letting $z^*\in Z_{A_i}^*$ be the smallest 
    such that $z^*>z$, one would also have $I_{F_0^i,G_i}(z^*)<0$,
    contradiction. Meanwhile, $F_i\in\mathcal{I}_{F_0^i}$, and satisfies
    the rest of the conditions of NBPS-M by construction.    
    \end{proof}
        
    Define the matrix $\mathbf{A}$ with 
    $\vert \mathcal{A}\vert \times (1+\vert Z\vert) +2\vert X\vert$ rows 
    (indexed by $i$, corresponding to menus, pairs of menus with 
    posterior means in the support of the prior, or actions with $0$ or 
    $1$) and $\sum_{A\in\mathcal{A}}[1+\sum_{a\in A}(1+\vert Z_{A,a}\vert)]$ 
    columns (indexed by $j$, corresponding to menus $A$, posterior 
    means $z_{A,a}$, or posterior means contained in $Z_{A,a}$). 
    The entries of $\mathbf{A}$ are given by 
    \[
    \mathbf{A}_{i,j}=\begin{cases}
        1, & i=(A,0), j=(A,a,z_{A,a})\\
        -1, &  i=(A,0); j=(A,a,z),z\in Z_{A,a}\\
        z^*-z_{A,a}, & i=(A,z^*), z^*\in Z_A^*\setminus\{0,1\}; j=(A,a,z_{A,a})\\
        z-z^*, & i=(A,z^*), z^*\in Z_A^*\setminus\{0,1\}; j=(A,a,z), z\in Z_{A,a}\\
        -1, & i=(A,z^*), z^*\in Z_A^*\setminus\{0,1\}; j=A\\
        1-z_{A,a}, & i=(A,1); j=(A,a,z_{A,a})\\
        z-1, & i=(A,1); j=(A,a,z),z\in Z_{A,a}\\
        z_{A,a}, & i=(a,1); j=(A,a,z_{A,a})\\
        -z, & i=(a,1); j=(A,a,z),z\in Z_{A,a}\\
        1, & i=(a,0); j=(A,a,z_{A,a})\\
        -1, & i=(a,0); j=(A,a,z),z\in Z_{A,a}\\
        1, & i=A; j=(A,a,z),z\in Z_{A,a}\\
        -1, & i=A; j=A\\
        0, & \mbox{otherwise}
    \end{cases}
    \]
    Meanwhile, define the vector $\mathbf{b}$ of length 
    $\sum_{A\in\mathcal{A}}[1+\sum_{a\in A}(1+\vert Z_{A,a}\vert)]$
    such that (using the same indexes $j$ as in $\mathbf{A}$)
    $$\mathbf{b}_j=\begin{cases}
    -1, & j=A\\
    0, & \mbox{otherwise}
    \end{cases}
    $$
    
    By applying Farkas' lemma (\cite{aliprantis2006infinite}, 
    Corollary 5.85), there exist (i) scalars $u_A^k$, for $a\in X$ 
    and $k\in\{0,1\}$, (ii) scalars $\lambda_{z,A}$ for $A\in\mathcal{A}$ 
    and $z\in Z^*_A$ such that $\lambda_{z,A}\geq 0$ for 
    $z\notin \{0,1\}$, and (iii) scalars $\gamma_A\geq 0$ for $A$ such 
    that $F_{\sigma_A}\neq \delta_{z_0}$, that satisfy:    
\begin{equation}
\label{farkas_f}
    \lambda_0^A+\sum_{z^*>z} [z^*-z]\lambda_{z^*,A}-z_{A,a} u^1_{a}-u_a^0\leq 0,\qquad\forall A, a\in \mbox{supp}(\sigma_A)
\end{equation}
\begin{equation}
\label{farkas_g}
    -\lambda_0^A-\sum_{z^*>z} [z^*-z]\lambda_{z^*,A}+zu_a^1+u_a^0\leq 0,\qquad\forall (A,a), z\in Z_{A,a}
\end{equation}
\begin{equation}
\label{farkas_z0}
    -\lambda_0^A-\sum_{z^*>z_0} [z^*-z_0]\lambda_{z^*,A} +\gamma_A^0+z_0 u_{a^*(z_0)}^1+u_{a^*(z_0)}^0\leq 0,\qquad\forall (A,a), \qquad F_{\sigma_A}\neq \delta_{z_0}
\end{equation}
\begin{equation}
\label{farkas_y}
    -\sum_{z^*\in Z^*_A\setminus\{0,1\}} \lambda_{z^*,A}-\gamma_A^0\leq -1,\qquad \forall A:\,F_{\sigma_A}\neq \delta_{z_0}
\end{equation}

Set the payoffs for each $a$ as $u(a,z)=u_a^0+zu_a^1$, and the 
price function for each $A$ as 
\[
\Lambda_A(z)=\lambda_0^A+\sum_{z^*\in Z^*_A} [z^*-z]\lambda_{z^*,A}\mathbf{1}[z^*\geq z]
\]
Note that $\lambda_A$ is convex, and affine on any intervals on which
$I_{F_0,F_{\sigma_A}}(z)>0$. By (\ref{farkas_f}) and (\ref{farkas_g}), 
it follows that $\Lambda_A(z_{A,a})=u(a,z_{A,a})$. 
By (\ref{farkas_g}), for all $z\in \cup_{a} Z_{A,a}$, 
$u(a,z)\leq \Lambda_A(z)$. By Lemma \ref{outmsuff}, it follows that 
$F_{\sigma_A}$is optimal for menu $A$, as there is always an optimal 
distribution whose support is contained in $\bigcup_{a} Z_{A,a}$. 
By (\ref{farkas_y}), for each $A$, at least one of $\lambda_{z^*,A}$ or 
$\gamma_A^0$ is strictly positive whenever 
$F_{\sigma_A}\neq \delta_{z_0}$. If one of the former is strictly 
positive, then $\Lambda_A$ is not affine; if one of the latter, then by 
(\ref{farkas_z0}), $u(a^*_A(z_0),z_0)<\Lambda_A(z_0)$. In 
both cases, the sender strictly benefits by persuasion by Lemma 
\ref{pmben}.

\end{document}